\documentclass[journal]{IEEEtran}
\usepackage{amsmath,amsfonts}
\usepackage{amssymb}
\usepackage{algorithmic}
\usepackage{algorithm}
\usepackage{array}
\usepackage[caption=false,font=normalsize,labelfont=sf,textfont=sf]{subfig}
\usepackage{textcomp}
\usepackage{stfloats}
\usepackage{url}
\usepackage{verbatim}
\usepackage{graphicx}
\usepackage{cite}
\usepackage{xcolor}
\usepackage{amsthm} 

\def\Pr{\mathbb{P}}
\newtheorem{lemma}{Lemma}
\newtheorem{prop}{Proposition}
\newtheorem{theorem}{Theorem}

\definecolor{red}{rgb}{0.0, 0.65, 0.58}

\newcommand{\matr}[1]{\mathbf{#1}}     
\newcommand{\vect}[1]{\mathbf{#1}}     
\newcommand{\Mod}[1]{\ (\mathrm{mod}\ #1)}

\graphicspath{ {Figures/} }

\begin{document}

\title{Zadoff-Chu Sequences for Chirp-Domain Communication: Diversity-Complexity Tradeoffs in Doubly Dispersive Channels}
\author{Rawan~Alghamdi,~\IEEEmembership{Student Member,~IEEE,}~
        Moahmed~Siala,~\IEEEmembership{Senior Member,~IEEE,}~ 
        Tareq Y.~Al-Naffouri,~\IEEEmembership{Fellow,~IEEE,}~ 
        and~Mohamed-Slim~Alouini~\IEEEmembership{Fellow,~IEEE}~ 

\thanks{R. Alghamdi, T. Y.~Al-Naffouri, and M.-S. Alouini are with the Electrical and Computer Engineering Program, Division of Computer, Electrical and Mathematical Sciences and Engineering (CEMSE), King Abdullah University of Science and Technology (KAUST), Thuwal, 23955-6900, Kingdom of Saudi Arabia
 (email: \{rawan.alghamdi,  tareq.alnaffouri, slim.alouini\}@kaust.edu.sa).  (Corresponding author: R. Alghamdi)  }
\thanks{M. Siala is with the Mediatron Laboratory, Higher School of Communication of Tunis (SUP’COM), University of Carthage, Tunis 2083, Tunisia (email: mohamed.siala@supcom.tn). } 
}
\markboth{Journal of \LaTeX\ Class Files,~Vol.~14, No.~8, August~2021}%
{Shell \MakeLowercase{\textit{et al.}}: A Sample Article Using IEEEtran.cls for IEEE Journals}


\maketitle

\begin{abstract}
Orthogonal frequency-division multiplexing (OFDM) combats multipath-induced time dispersion by dividing the channel into narrowband sub-channels. However, when the channel also exhibits frequency dispersion due to mobility (Doppler effect), these sub-channels lose orthogonality and cause inter-carrier interference that degrades the reliability performance of the communication system. 
We investigate Zadoff-Chu (ZC) sequences for chirp-domain communication to improve reliability in time-frequency dispersive channels.
We show that ZC sequences are the only constant-amplitude zero-autocorrelation (CAZAC) sequences that transform a doubly dispersive channel into a singly dispersive channel that is either pure time or frequency dispersion. This transformation is controlled by the ZC root, which provides a geometric projection from the delay-Doppler domain onto a one-dimensional chirp-domain axis with a closed-form design rule. 
Because the transformed channel is singly dispersive, the receiver equalizes a one-dimensional convolutional channel rather than a two-dimensional delay–Doppler channel and can reuse trellis-based detectors that generate the soft information that coded systems require.
Then, we present ZC-based modulations, derive the effective channel after transformation, and analyze diversity, which reveals an underlying trade-off between diversity and receiver complexity. 
When evaluated at a vehicle speed of $540$ km/h and a carrier frequency of $4$ GHz, ZC-based modulations demonstrate performance comparable to affine frequency division multiplexing (AFDM) and orthogonal time-frequency space (OTFS), with gains of about  $5$ dB over orthogonal chirp division multiplexing (OCDM) and  $10$ dB over OFDM.
\end{abstract}


\section{Introduction}\label{secrel}
\IEEEPARstart{M}{ulti-carrier} techniques, which superimpose several carrier-modulated waveforms, are one of the most widely adopted physical layer techniques to combat the adversity of the wireless channel \cite{2014sahin_mc}. The most well-known and understood multi-carrier waveform scheme is orthogonal frequency-division multiplexing (OFDM), which can combat the time-dispersion in linear time-invariant (LTI) systems \cite{matz_TF}. The popularity of this technique stems from the fact that OFDM systems modulate the information symbols with base functions that correspond to the eigenfunctions of the LTI channel, which are well-approximated by Fourier series. Unfortunately, this elegant solution breaks down when the channel simultaneously experiences time and frequency dispersion because of the time-variance nature of the channel. The eigenfunctions of linear time-variant (LTV) channels cannot be expressed as complex exponentials, and no analytical expression exists for the eigenfunctions of a general LTV system \cite{hlawatsch2011_book}.
The study of reliable communication in doubly selective channels becomes particularly interesting in the context of next-generation wireless systems. One reason is that high mobility induces Doppler spread, which renders the channel time and frequency dispersive. In addition, next-generation wireless communication systems are expected to operate in high-frequency bands to support large bandwidth, which may be challenging since Doppler shifts are proportional to the frequency. Since high-mobility communication is needed to support several pillars of 6G networks---such as high-speed trains, low-Earth-orbit (LEO) satellites, airplanes/drones, and mmWave high-frequency spectrum band operations---reliable communication over doubly selective channels is a necessary condition for the future of wireless communication networks.

Various solutions have aimed to tackle the simultaneous time and frequency dispersions in wireless communication channels that are often found in high-mobility systems. Recently, several multi-carrier waveforms have been proposed to tackle the challenges of time and frequency variance from a physical-layer perspective, such as frequency-modulated OFDM (FM-OFDM) \cite{2023Hernando_FMOFDM}, orthogonal time frequency space (OTFS) \cite{Hadani2017_otfsm}, orthogonal chirp division multiplexing (OCDM) \cite{Ouyang2016_OCDM}, and affine frequency division multiplexing (AFDM) \cite{2021_Benami}
modulations. FM-OFDM modulates the information onto the phase difference, which makes this technique effective in Doppler spread channels, although at the cost of reduced spectrum efficiency \cite{2023Hernando_FMOFDM}. OTFS modulation can combat the degradation of the doubly selective channel by spreading the information symbols onto the delay-Doppler orthogonal basis functions. However, the ideal pulse for an OTFS system must satisfy the bi-orthogonality condition---such a pulse does not exist due to the uncertainty principle \cite{Lin2022_MCODDM}. Moreover, OTFS modulation relies on block-based detection, which increases the detection latency \cite{2023Yuan_otfsSurvey}.

\IEEEpubidadjcol 

Another approach to waveform modulation uses chirp signals, which are complex exponential signals with linearly time-varying instantaneous frequencies. One of the earliest uses of chirp signals to mitigate channel dispersion was reported in 2001 \cite{martone2001_fracfourier}, in which the chirp signals are generated via fractional Fourier transform (FrFT). Using a similar concept, Ouyang and Zhao proposed OCDM waveforms for linear time-invariant (LTI) systems in 2016  \cite{Ouyang2016_OCDM}. OCDM waveforms use the Fresnel transform to modulate information onto orthogonal complex linear chirps. 
Recent works in the literature show that OCDM modulation can mitigate Doppler frequency shifts in doubly selective channels. Specifically, \cite{omar2021_ocdm} provided a numerical analysis of the performance of OCDM in different impaired wireless channels, including LTI and LTV channels. More recently, \cite{liu2024_ocdmMP, 2024_haifOCDM} analytically evaluated the performance of OCDM in doubly selective channels, and showed that OCDM-modulated signals may still experience chirp fading in some doubly selective channels, where two or more chirps overlap and become inseparable. Thus, high-complexity estimation and detection algorithms are required to address this issue.

AFDM is another chirp-based waveform for doubly dispersive channels \cite{2021_Benami, rou2024_mag}, which generalizes OCDM by introducing two tunable chirp parameters through the discrete affine Fourier transform (DAFT). Specifically, while OCDM modulates information symbols using the Fresnel transform with a fixed chirp rate, AFDM employs the DAFT basis functions with chirp parameters, commonly denoted as $c_1$ and $c_2$. By appropriately selecting these parameters based on the channel's delay-Doppler profile, AFDM ensures that all propagation paths are separable in the DAFT domain. Moreover, the work in \cite{2023_Benami} shows how to optimize the parameters to achieve full diversity in doubly selective channels. In this regard, both AFDM and OCDM exploit chirp-domain representations of the channel; however, AFDM addresses the chirp-domain fading limitation of OCDM by continuously optimizing the DAFT basis. The fundamental question of which sequence families possess the algebraic structure to transform a doubly dispersive channel into a singly dispersive one, and whether such sequences are unique, remains open. Answering this question is important because it provides theoretical insight into the design space of chirp-based modulations and reveals the structural properties that underpin the effectiveness of chirp-domain communication in time-frequency dispersive channels.

In this paper, we propose a novel chirp modulation based on Zadoff-Chu (ZC) sequences. We show that properly designing the chirp rate/root of the ZC sequence circumvents the chirp fading that is prevalent in OCDM systems. Moreover, we show that OCDM is a trivial case of the proposed ZC modulation. Thus, ZC modulation provides a design parameter to avoid selectivity in the chirp domain. 
Moreover, 
we demonstrate that ZC sequences exhibit a unique relationship with doubly selective channels that transforms the time and frequency dispersion into either a pure time dispersion or a pure frequency dispersion. 
This relationship enables the use of classical detection techniques for LTV systems, such as a maximum-likelihood sequence equalizer (MLSE) using Viterbi and Trellis-based detectors and linear minimum mean square error (LMMSE).  

The contributions of this article are as follows:

\begin{itemize}
\item We provide a unified theoretical framework of chirp signaling in doubly dispersive channels.
\item We introduce a new chirp modulation technique using ZC sequences, study its input-output relationship for doubly selective channels, and show that OCDM is a special (trivial) case of ZC modulations. 
\item We analyze the unique properties of ZC sequences in doubly selective channels. Specifically, we analytically show that modulating a signal using ZC sequences can transform the doubly dispersive channel effect into a pure dispersion channel in the chirp domain, and prove that ZC sequences are the only sequences in the constant amplitude zero autocorrelation (CAZAC) waveform family that have this property.

\item We highlight the key differences in using ZC over OCDM in terms of resolving the overlap of time-domain channel paths with distinct delay-Doppler profiles in the chirp domain, provide an analytical framework for designing ZC sequences for various wireless communication channels, and study the performance limits under practical considerations. 

\item We derive the diversity order of the proposed modulation techniques using ZC sequences for LTV channels.

\item We show that the transformed channel acts as a one-dimensional convolutional filter, so the receiver can reuse classical trellis-based detectors, such as the Viterbi algorithm, the Bahl-Cocke-Jelinek-Raviv (BCJR) algorithm (also known as the maximum a posteriori (MAP) algorithm), the log-MAP algorithm (the logarithmic domain implementation of the BCJR algorithm), the max-log-MAP algorithm (a reduced-complexity approximation of the log-MAP algorithm), and the soft output Viterbi algorithm (SOVA), rather than solving a two-dimensional delay-Doppler equalization problem.

\end{itemize}

This article is organized as follows. We study the fundamental challenges of doubly selective channels and review OCDM modulation in Section \ref{sec:ocdm}. We then introduce the ZC sequences for LTV channels in Section \ref{sec:zc}; we specifically study the system design and configuration, explore the favorable properties of ZC sequences for tackling the time-varying channels, and analyze the diversity order of the proposed modulation for the LTV channel. Finally, Section \ref{sec:results} provides the numerical results and further discusses the use of classical equalization methods in the chirp domain, as well as demonstrating the performance gap between OCDM and ZC modulation.

\section {Orthogonal Chirp Division Multiplexing in Linear-Time Variant Systems}
\label{sec:ocdm}
\noindent This section establishes the channel model for linear time-varying systems, with a specific focus on the effects of time and frequency dispersion, and then revisits the literature on the OCDM techniques used to mitigate dispersive channels.

\subsection{Linear Time-Variant Systems and Doubly Dispersive Channels}\label{sec:ltv}

Wireless communication channels are characterized by time-frequency variations of the channel strength \cite{Tse_Viswanath_2005}, resulting from multipath propagation and Doppler effects. Multipath propagation can lead to significant distortion of the received signal, causing signal dispersion in time. Similarly, movement of the transmitter, receiver, or objects in the environment causes the received signal to experience Doppler effects, causing frequency dispersion of the signal. Channels with frequency and time dispersions are referred to as doubly dispersive channels; such channels pose severe challenges for reliable communication systems.

The received continuous-time signal, $r(t)$, resulting from transmitting the signal $s(t)$ through a (noiseless) doubly dispersive channel of $L$ separable propagation paths, is a superposition of time- and frequency-shifted copies of the transmitted signal produced by scattering from reflective objects in the propagation environment. 
Each of these $L$ propagation paths can be characterized by several physical attributes, such as the time delay, $\tau_i$, Doppler frequency shift,  $\nu_i$, and complex fading gain, $h_i$,  $i\in\{0,\ldots, L-1\}$ \cite{matz_TF}. Therefore, the received continuous-time signal, $r(t)$, is given by 
 
\begin{equation}\label{eq:y_basic_time}
r(t)= \sum_{i=0}^{L-1} {h_i}   s\left(t - \tau_i\right)e^{j2\pi \nu_i \left(t - \tau_i\right)}  
\end{equation}
Here, we assume that the time shift occurs first, followed by the frequency shift. One can easily find an equivalent interpretation of the case where the frequency shift occurs first, followed by the time shift, which results in a phase difference, expressed as $e^{j2\pi \nu_i \tau_i}$, that can be incorporated in the complex fading gain, $h_i$.  From \eqref{eq:y_basic_time}, we can deduce the time-delay domain representation of the channel impulse response as 
\begin{equation}
h(t, \tau) = \sum_{i=0}^{L-1} {h_i}   \delta \left[\tau - \tau_i\right] e^{j2\pi \nu_i \left(t- \tau_i\right)} , \label{eq:cir_td}
\end{equation}
where $\delta[\cdot]$ denotes the Dirac delta function and $\tau$ is the delay domain variable. Hence, we can write the input-output relationship between $s(t)$ and $r(t)$ as
\begin{equation}
r(t)= \int_{-\infty}^{\infty}  s(t- \tau) h(t, \tau) d\tau. 
\end{equation}

It is often more convenient to work with the baseband discrete-time representation of the communication system. Hence, we assume the total bandwidth of $s(t)$ is $B$ with a frame duration of $T$ seconds, such that we have $N = BT$ complex samples per frame. Then, we take the sampling period as $T_s = \frac TN$. Hence, we sample $r(t)$ at time instants $t = k T_s, k=0,1,\ldots ,N-1$, to obtain the discrete-time received signal $r[k]$ as

\begin{equation} \label{eq:y_basic}
r[k] = \sum_{i=0}^{L-1} {h_i}   s\left[k - l_i\right]e^{j2\pi \kappa_i \left(k-l_i\right)}, 
\end{equation}
where $l_i = \frac{\tau_i}{T_s}$ and  $\kappa_i = \frac{\nu_i }{T_s}$ are the normalized delay and Doppler shifts. We assume that the number of samples is sufficiently large that there is no effect of fractional delay and Doppler on the performance, i.e.,  $l_i, \kappa_i \in \mathbb{Z}$.

Importantly, by appending a cyclic prefix (CP) to the transmitted signal, $s[k]$, the received signal in the concerned frame duration, $r[k]$, appears to be periodic of period $N$, and the channel operation can be simplified to cyclic convolution. Hence, we take the length of the CP to be larger than or equal to the total delay spread of the channel. 
We can also express the channel in \eqref{eq:cir_td} in the time domain in matrix form, as 
\begin{equation}\label{eq:channelMatrix}
\matr{H} = \sum_{i=0}^{L-1} {h_i}\matr{\Pi}^{l_i} \matr{\Delta}^{\kappa_i},
\end{equation}
where $\matr{\Pi}$, given by
\begin{equation}
{\matr{\Pi }} = {\left[ {\begin{array}{*{20}{c}}
0& \cdots &0&1\\
1& \ddots &0&0\\
 \vdots & \ddots & \ddots & \vdots \\
0& \cdots &1&0
\end{array}} \right]_{N \times N}},   \label{C4_TD_forward_cyclic_shift_matrix}
\end{equation}
is a circulant permutation matrix, equivalent to a forward cyclic shift by one sample in the time domain. In other words, $\matr{\Pi}^{l_i} $ operates a shift of $l_i$ samples in the time domain. Moreover,  $\matr{\Delta}=\textrm{diag}\{{\gamma}^0,{\gamma}^1,\ldots,{\gamma}^{N-1}\} $ is a diagonal matrix with $\gamma \buildrel \Delta \over = e^{\frac{j2\pi }{N}}$, and the operator $\textrm{diag}\{\vect{u}\}$ refers to the diagonal matrix with vector $\vect{u}$ on the main diagonal. The matrix $\matr{\Delta}$ operates a unit frequency shift in the frequency domain, and hence, $\matr{\Delta}^{\kappa_i}$ operates a shift of $\kappa_i$ units in the frequency domain.
Then, we can express the input-output relationship in the time domain as 
\begin{equation} \label{eq:timedomain_r}
\vect{r} = \matr{H} \vect{s}.
\end{equation}

Now that we have established the mathematical modeling of the doubly-dispersive channel, we next study OCDM modulation for linear time-variant systems.

\subsection{OCDM Modulation for LTV}\label{sec:ocdm_channel}
OCDM is a particular case of spread spectrum techniques that leverages the orthogonality of cyclically shifted linear chirp signals to maximize spectral efficiency at the Nyquist signaling rate \cite{Ouyang2016_OCDM}. 
Rather than modulating the signal on parallel frequency resources as in OFDM, OCDM technology synthesizes a large number of linearly frequency-modulated (LFM) (i.e., chirped) waveforms to modulate the signal \cite{ouyang2024_ce}. 

The theory of OCDM, as established in  \cite{Ouyang2016_OCDM}, depends on the Fresnel transform, which is ``the formula that mathematically describes the near-field optical diffraction." 
The discrete Fresnel transform (DFnT) matrix, 
 $\matr{\Phi}$, of size $N \times N$ is defined as  
\begin{equation}\label{eq:fdnt}
{\Phi}[m,n] = \frac{1}{\sqrt{N}} e^{-j \frac \pi 4}
\begin{cases}
 e^{j \frac{\pi (m-n)^2}{N}}, 		& N ~\text{even},\\
 e^{j \frac{\pi (m-n+\frac 12)^2}{N}}, 	& N~ \text{odd}.
 \end{cases}
\end{equation}
The DFnT matrix $\matr{\Phi}$ can be decomposed as $\matr{\Phi} = \matr{\Lambda} \matr{F} \matr{\Lambda}$, where $\matr{\Lambda} = \mathrm{diag}(e^{j\pi k^2 / N})$ and $\matr{F}$ is the $N$-point discrete Fourier transform (DFT) matrix. This decomposition reveals that the DFnT consists of a chirp pre-multiplication, a Fourier transform, and a chirp post-multiplication.
Since the DFnT matrix is unitary and circulant, the inverse DFnT (IDFnT) can be easily derived \cite{Ouyang2016_DFnT}. Moreover, we can diagonalize the DFnT matrix by the DFT matrix, which leads to a simple implementation of the DFnT \cite{Ouyang2016_OCDM}. 

To obtain the OCDM-modulated transmit signal, $\vect{s}$, we precode the data symbol vector, $\vect{x}$, with the IDFnT, to get
\begin{equation} \label{eq:timedomain_s}
\vect{s} =  \matr{\Phi}^H \vect{x}.
\end{equation}
Similarly, at the receiver, we obtain the discrete OCDM signal, $\vect{y}$, resulting from taking the DFnT of the discrete received signal, $\vect{r}$, of length $N$, as 
\begin{equation}\label{eq:chirpdomain_y}
\vect{y} = \matr{\Phi} \vect{r}.
\end{equation}

\noindent  By collecting \eqref{eq:timedomain_r}, \eqref{eq:chirpdomain_y} and \eqref{eq:timedomain_s}, the input-output relationship of the OCDM system is
\begin{equation}
\vect{y} = \matr{\Phi} \matr{H}  \matr{\Phi}^H \vect{x} = \matr{H}_c \vect{x},
\end{equation}
where $\matr{H}_c = \matr{\Phi} \matr{H}  \matr{\Phi}^H$. Here, $\matr{H}_c$ is the effective channel matrix in the chirp domain. We obtain a better understanding of the effects of modulating data symbols with OCDM by deriving the expression for $\matr{H}_c$, as detailed in Appendix \ref{apen:chirp_channel}. Specifically, we get  
\begin{equation}
\begin{split}   
&{H}_c \left[m,n\right] =\\
						&  \sum_{i=0}^{L-1} h_i e^{j \frac{\pi}{N}( 2m\kappa_i - 2l_i\kappa_i - \kappa_i^2)} 
						\delta[n-m+l_i +\kappa_i \Mod{N}], 
\end{split}
\label{eq:chirp_channel} \end{equation} 

\noindent where time-shifts are to be understood as modulo $N$. 
From the expression of ${H}_c \left[m,n\right]$ in \eqref{eq:chirp_channel}, we can deduce that the channel effect in the chirp domain is a complex fading with a coefficient $\tilde{h}_i = h_i e^{-j \frac{\pi}{N}( 2l_i\kappa_i + \kappa_i^2)} $, a delay of $l_i+\kappa_i$ units, and a frequency shift of $\kappa_i$. That is, a shift in the chirp domain is comprised of $\kappa_i$ units due to Doppler, and $l_i$ units due to the multipath time-shift. 
The expression of ${H}_c \left[m,n\right]$ reveals important information about the impact of the doubly-dispersive channel in the chirp domain using OCDM modulation. Namely, 
there is a possibility that two or more paths overlap in the chirp domain, i.e.,  $l_i+\kappa_i = l_j+\kappa_j$ for $i\neq j$ \cite{2024_haifOCDM, liu2024_ocdmMP}. Therefore, the performance of the OCDM modulation may degrade in some doubly dispersive channels due to chirp fading, where some of the delay-Doppler paths are unidentifiable (inseparable) in the chirp domain. \cite{2024_haifOCDM}.  

Next, we show that achieving identifiability (separability) in the chirp domain is possible by using ZC sequences to modulate information symbols.

\section{Zadoff-Chu modulation in Linear-Time Variant Systems}\label{sec:zc}
\noindent  This section develops the ZC modulation by using ZC sequences to modulate the information-bearing symbols to achieve better separability in the chirp domain.

\subsection{Key Features of Zadoff-Chu Sequences}
ZC sequences are complex sequences with a unit amplitude and polyphase  \cite{andrews2023_zc}.  
One key property of ZC sequences is that cyclically shifted versions of the sequences are orthogonal to one another, i.e., zero auto-correlation. A ZC sequence has two parameters: the sequence length, $N$, and the root index, $R, 1\le R <N$. The parameter $R$ must be in the multiplicative group of integers modulo $N$, that is, $R$ and $N$ must be co-prime.  Given the pair of parameters $(R,N)$, the ZC sequence is defined as 
\begin{equation}
\label{eq:defn}
z_R [k]  = \begin{cases}
			\exp\left[ -j\pi R \frac{k(k+1)}{N} \right], & \text{if $N$ is odd},\\
            		\exp\left[ -j\pi R \frac{k^2}{N} \right], & \text{if $N$ is even},
 \end{cases}
\end{equation}
where $k  = 0, 1,2, \ldots,N - 1$.   %

In addition to its CAZAC properties, a ZC sequence can transform a shift in the time domain onto a frequency shift in the chirp domain up to a unit modulus complex factor, and vice versa. 
Therefore, in this work, we claim that if a signal modulated by $z_R [k]$ is shifted in time, it is equivalent to the same signal, but with a frequency shift. Similarly, if the ZC-modulated signal is shifted in frequency, it is equivalent to the same signal with a time shift.
This property of the ZC-modulated signals is attractive in doubly selective channels because the signal's frequency and time shift can be transformed into either a pure time shift or a frequency shift. In other words, suppose  $x_R\left[k\right]$ is a ZC-modulated signal of length $N$. 
Then, when  $x_R\left[k\right]$ experiences a Doppler (frequency) shift by a multiple of $1/N$, it translates to a time shift of $P \in \{0,1,\ldots, N-1\}$ units. In Section \ref{sec:designR}, we show that the design of $R$ depends on $P$. Therefore, we have the following relation 
\begin{equation}\label{ZC_mod_time}
x_R\left[k\right]e^{\frac{j2\pi k}{N}} = \alpha x_R[k-P \Mod N], 
\end{equation}
\noindent where $\alpha$ is some complex multiplicative factor of unit modulus. Similarly, a unit time shift of $x_R\left[k\right]$ translates into a frequency shift of $\frac vN, v \in \{0,1,\ldots,N-1\}$ that depends intimately on $R$, up to a complex multiplicative factor of unit modulus $\alpha_f$. That is, 

\begin{equation}\label{ZC_mod_freq}
x_R\left[k - 1 \Mod N \right]= \alpha_f x_R[k]e^{\frac{j2\pi kv}{N}} . 
\end{equation}
The ZC-modulated signal in \eqref{ZC_mod_time} translates the two-dimensional shift into only a time shift; thus, the channel is equivalent to a multipath channel without Doppler components. The signal in \eqref{ZC_mod_freq} translates time and frequency shifts into only a frequency shift, resulting in a channel of Doppler shifts without delay components. Since the DFT dual of a ZC sequence is still a ZC sequence \cite{beyme2009_dftzc}, we have an equivalent multipath channel in the frequency domain. In what follows, we show that ZC sequences are the only CAZAC sequences that give the relationship in \eqref{ZC_mod_time} and \eqref{ZC_mod_freq}.

\subsection{Characterizing CAZAC sequences}
In this section, we prove that ZC sequences are the only finite-length CAZAC sequences that satisfy the relationship in \eqref{ZC_mod_time} and \eqref{ZC_mod_freq}. For the sake of rigor, we prove each relationship separately. However, one may choose to prove one case and use the Fourier transform and its time-frequency shift characteristic to prove the other. Moreover, starting from the relationship in  \eqref{ZC_mod_time} and then using the Fourier transform to prove \eqref{ZC_mod_freq} may lead to a more involved proof than the inverse approach. 
We start by proving the case where every frequency shift is translated into a time shift, as in  \eqref{ZC_mod_time}.

Let $\left\{a[k]\right\}_{k=0}^{N-1}$ be a finite-length cyclically CAZAC sequence of length $N$, that if modulated by a frequency shift of the form $k/N, k \in \{1,\ldots, N-1\}$ is equivalent, up to a multiplicative factor of unit modulus, $\alpha$, to the same sequence cyclically shifted in time. 
 In particular, for a frequency shift $\frac 1N$, the resulting sequence is shifted by $\tau \Mod N, \tau \in \mathbb{Z}_{\neq 0}$ samples in time, i.e.,
\begin{equation}
a[k] e^{\frac{j 2 \pi k}{N}}=\alpha a[k-\tau \Mod N], \quad 0 \leq k<N. \label{eq:freqtotime}
\end{equation}
The projection of frequency shifts onto time shifts can be generalized as follows: 
\begin{lemma}\label{lem:mfreqshift}
Given \eqref{eq:freqtotime}, where an elementary frequency unit shift leads to a time shift of $\tau  \Mod N$ units, then, a frequency shift of $\frac mN, m \in \mathbb{Z}_{\neq 0}$ leads to a time shift of $m\tau  \Mod N $-units and a constant multiplicative phase shift by $ \alpha^m e^{j2\pi\frac{m(m-1)\tau}{2N}}$. In other words, we obtain
\begin{equation}
a[k] e^{\frac{j 2 \pi mk }{ N}}=\alpha^m a[k-m\tau\Mod N] e^{j2\pi\frac{m(m-1)\tau}{2N}}. \label{eq:seq_m} \end{equation}

\end{lemma}
\begin{proof}
 We use proof by induction. The base case for $m = 1$ is easily verifiable. 
Now, we form the hypothesis that for $m= n$, we get 
\begin{equation}\label{eq:hypoth}
a[k] e^{\frac{j 2 \pi nk }{ N}}=\alpha^n a[k-n\tau \Mod N] e^{j2\pi\frac{n(n-1)\tau}{2N}}. 
\end{equation}

In the induction step, we work with $m= n+1$. Our goal is to prove that 
\begin{equation}
a[k] e^{\frac{j 2 \pi (n+1)k }{ N}}=\alpha^{(n+1)} a[k-(n+1)\tau \Mod N] e^{j2\pi\frac{n(n+1)\tau}{2N}}. 
\end{equation}

We start with our hypothesis \eqref{eq:hypoth}, and multiply both sides of \eqref{eq:hypoth} by a frequency shift of $\frac 1N$, that is
\begin{align}
a[k] e^{\frac{j 2 \pi { (n+1)}k}{N}} &=\alpha^n a[k-n\tau \Mod N] e^{j2\pi\frac{n(n-1)\tau}{2N}} {e^{\frac{j 2 \pi k}{N}}}.\label{eq:induc}
\end{align}

Here, we digress to revisit the base case in \eqref{eq:freqtotime} and substitute with the following relationship $k\to k-n\tau$, and by isolating the term $a[k-n\tau]$ on one side, we obtain the desired result as
\begin{align}
a[{k-n\tau} \Mod N] &=\alpha a[k-(n+1)\tau\Mod N] {e^{\frac{-j 2 \pi ({k-n\tau}) }{ N}}}. \label{eq:step} 
\end{align}

Now, we can use \eqref{eq:step} in \eqref{eq:induc} to get 
\begin{align}
a[k] e^{\frac{j 2 \pi  (n+1)k}{N}} &=\alpha^{n+1}  a[k-(n+1)\tau \Mod N]  e^{\frac{j  \pi { n(n+1)\tau}}{N}}.
\end{align}

\end{proof}

It is easy to see that if $m<0$, then \eqref{eq:seq_m} becomes 
\begin{equation}
a[k] e^{-\frac{j 2 \pi mk }{ N}}=\tilde{\alpha}^m a[k+m\tau\Mod N] e^{j\pi\frac{m(m+1)\tau}{N}}, \label{eq:negm}
\end{equation}
where $\tilde{\alpha} = \alpha^{-1}$.

The conclusion we can draw here is that if a pure frequency shift of $\frac 1N$ leads to a $\tau$-unit time shifts, the $\frac mN$ frequency shift will lead to a $m\tau$-unit time shift and a constant multiplicative factor $\alpha^m e^{j2\pi\frac{m(m-1)\tau}{2N}}$.

Since all time shifts occur modulo $N$, then there exists a frequency shift of $\frac bN$ such that the time shift $b\tau \bmod N$ is the least positive non-zero time shift, which we denote for ease of notation using $d$. Then, we make the following proposition:

\begin{prop}\label{prop:leastfreq}
 Any time shift induced by a frequency shift must be an integer multiple of $d$. 
\end{prop}
\begin{proof} Since  Lemma \ref{lem:mfreqshift}  shows that any integer frequency shift systematically leads to an integer time shift that is a multiple of $\tau$, it suffices to show that $\tau$ is an integer multiple of $d$. In other words, if we express $\tau$ in terms of the Euclidean division, as 
 \begin{equation}
\tau = qd +r, \quad q,  r\in \mathbb{Z}, \quad 0\leq r<d,
\end{equation}
then, we need to show that $r$ must be zero. 

 Suppose we modulate the sequence by $\frac{1-qb}{N}$, which induces a time shift of $\tau - qd \bmod N$, which --- by definition --- is equal to $r$. If $r\neq0$, then we have a possible positive time shift that is strictly lower than the minimum shift $d$, which leads to a contradiction. Hence, $r$ must be zero, and we conclude that $\tau$ is indeed a multiple of $d$. 
 \end{proof}
 
It follows from Lemma \ref{prop:leastfreq} that all permissible shifts in time are integer multiples of $d$. Hence, we conclude that $N$ is also an integer multiple of $d$ because, otherwise, there is a frequency shift of $\frac{pd}{N}$, where $p$ is an integer, resulting in a time shift of $pd \Mod{N}$ which is non-zero and strictly lower than $d$. This contradicts the fact that $d$ is the minimum non-zero time shift. 

Since $d$ divides $N$, then $p \in \{0, 1, \ldots, \frac{N}{d}-1\}$ is sufficient to cover all cases. In fact, we can take $N = Qd$, where $Q$ is some integer, and hence it is sufficient to know the values of $a[0], a[1], \ldots, a[d-1]$ to characterize the full sequence, $a[k] $. Therefore, the rest of the sequence is simply a multiplicative factor of the first $d$ entries of the sequence, as dictated by the properties of time and frequency shift in \eqref{eq:freqtotime}. To reiterate the implication of Proposition \ref{prop:leastfreq}, the general form of a shift of $\frac bN$ is
\begin{equation}
\begin{split}
 a[k]  e^{\frac{j 2 \pi bk }{ N}} &=\alpha^b e^{j\pi\frac{(b-1)d}{N}} a[k-d\Mod N] \\&=\gamma a[k-d\Mod N], 
 \end{split}
 \label{eq:bshift}
 \end{equation}
where $\gamma=\alpha^b e^{j\pi\frac{(b-1)d}{N}}$ is a multiplicative factor independent of $k$. Then, for a given $l, 0\leq l <d$, we can generalize the expression for a  frequency shift of $-p \frac bN$, where $p$ is an integer, and we have 
  \begin{equation}
  \begin{split}
a[l+pd\Mod N]  &= a[l]  e^{-\frac{j 2 \pi pbl }{ N}} \gamma^{p} e^{-j\pi\frac{bd p(p+1)}{N}}\\
			& =  a[l]  e^{-\frac{j 2 \pi pbl }{ N}} \chi^{p} e^{-j\pi\frac{bd p(p-1)}{N}},
\end{split} \label{eq:minuspbshift}
 \end{equation}
where $\chi^{p}  = \alpha^{pb} e^{-j\pi\frac{(b+1)pd}{N}}$. 

In other words, \eqref{eq:minuspbshift} shows that knowing the first $d$ entries of $\vect{a}$ reveals its entire structure. 

\noindent To characterize the subsequence $a[l], 0 \leq l<d$, we start by finding the value of $d$.

\begin{lemma}\label{lem:value_d}
The least positive non-zero time shift, $d$, in the modulo $N$ sense for a sequence that projects a frequency shift into a time-shifted version of the same sequence scaled by a multiplicative factor must be one. 
\end{lemma}
\begin{proof}

As a result of Proposition \ref{prop:leastfreq}, we take $p = Q$ in  \eqref{eq:minuspbshift}, which amounts to a shift of $\frac {Qb}{N}=\frac{b}{d}$ in frequency, and that has the effect a shift in time by $pd =Qd$, and leads to

  \begin{equation}
 a[l+Qd\Mod N]  = a[l] e^{-\frac{j 2 \pi l b }{d}}e^{-j\pi b(Q-1)}  \chi^{Q} 
  \end{equation}
 which, given that $Qd=N$, that indices are taken modulo $N$, we can rearrange to obtain
\begin{equation}
a[l] \left( 1- \eta^Q (-1)^{bQ}  e^{-j\frac{2\pi bl}{d}}  \right)=0,  \label{eq:s_l}
\end{equation}
where $0 \leq l<d$, and we can again group the multiplicative factor as $\eta^Q =(\chi  e^{j\pi b/Q})^Q$, since it does not depend on $l$. 

From \eqref{eq:s_l}, we have three possible cases. The first case is where $a[l]=0, 0 \leq l<d$, which suggests that the sequence is null, and thus we omit this trivial case from our treatment. 
In the next two cases, we assume that  $a[l], 0 \leq l<d$ is not  null; hence, we can re-write \eqref{eq:s_l} as
\begin{equation}
 \left(\eta (-1)^{b} \right)^Q =e^{j\frac{2\pi bl}{d}}.  \label{eq:s_l2}
\end{equation}
The second case is where only one sample of $a[l]$ is non-zero, i.e., there exists a unique value of $l$ such that $0 \leq l<d$, and  $a[l]\neq 0$. This case is equivalent to having $d=1$, since we can down-sample the sequence $a[k]$ by a factor $d$. That is, we keep the non-zero samples in $a[k]$, which are regularly spaced by $d$. In the third case, we assume that at least two or more samples of $a[l]$ are non-zero and, subsequently, $d\geq2$; for example, $l_1$ and $l_2$, where $0 \leq l_1<l_2<d$. 
Then, the right-hand side of \eqref{eq:s_l2}  does not depend on $l$ for those samples $a[l], 0\leq l<d$, for which $a[l]\neq0$. 
This means that if $l_{1}$ and $l_{2}$ are two of these indices, then $b\left(l_{2}-l_{1}\right)$ needs to be divisible by $d$, which is not possible, as we will demonstrate next. The objective then is to show that the only valid case is that there exists only one unique sample that is non-null, which, by downsampling by $d$, is equivalent to having $d=1$.

We start the proof by contradiction by equating the modulation terms of $a[l+pd\Mod N]$ in \eqref{eq:minuspbshift} with \eqref{eq:freqtotime} where $k \rightarrow l+pd\Mod N$. 

\begin{equation}
 e^{-j\frac{2\pi l}{N}} =e^{-\frac{j 2 \pi pbl }{ N}}.
\end{equation}

 By observing the exponent terms, we obtain

\begin{equation}
\frac{2\pi }{N} l \left(bp - 1\right) = c \Mod{2\pi}, \quad  0\leq l<d. \label{eq:mod2pi}
\end{equation}
If we take the difference between the two indices, $l_1, l_2$, substituted in \eqref{eq:mod2pi}, we get 

\begin{equation}
\frac{2\pi }{d} b\left(l_2-l_1\right)=0 \Mod{2\pi}, 
\end{equation}
which indicates that there exists an integer $v$, such that 
\begin{equation}
b \left(l_2-l_1\right)= v. \label{eq:pbminus1}
\end{equation}
 By taking the modulo $d$ of both sides of the previous equation, we end up with
\begin{equation}
l_{2}-l_{1}=0 \Mod{ N}.
\end{equation}
This equality is in contradiction with the fact that $0 \leq l_{1}<l_{2}<d$, and $0<l_{2}-l_{1}<d$.
 Therefore, we conclude that $d$ must be equal to 1.
\end{proof}

Thus far, we have shown the general form of a sequence, $a[k]$, that projects a frequency shift into a time shift, and such a sequence can be fully determined by $a[0], \ldots, a[d-1]$, and since $d$ must be $1$, therefore, the sequence, $a[k]$, is fully determined by $a[0]$. Now, we move to find $a[0]$ such that  $a[k]$ is a CAZAC sequence. In other words, the circular autocorrelation of $a[k]$, $\Gamma_v =\sum_{k=0}^{N-1} a[k]^{\ast} a[k+v]$, is zero for all $v \neq 0, i \in \mathbb{Z}, 0< v< N$.

\begin{theorem}\label{Th:freqShift}
The ZC sequences are the only family of  CAZAC sequences with the relationship in  \eqref{ZC_mod_time}.

\end{theorem}
\begin{proof}

From Lemma \ref{lem:value_d}, we know $d=1$, and the sequence we are searching for to satisfy  \eqref{ZC_mod_time} is fully determined
by $a[0]$
in \eqref{eq:minuspbshift}, which can be re-written as 
\begin{equation}
 a[p\bmod N]  =\chi^{p} e^{-j\pi\frac{b p(p-1)}{N}} a[0]   =  \left( \chi  e^{j\frac{\pi b}{N}} \right)^p e^{-j\pi\frac{b p^2}{N}} a[0]. \label{eq:s_p}
\end{equation}
 Since we are looking at an $N$-periodic sequence, then we must guarantee that $a[0]=a[N]$, meaning that 
\begin{equation}
\begin{split}
 a[0] = a[N] &= \left(  \chi e^{j\frac{\pi b}{N}}  \right)^{N} e^{-j\pi b N} a[0] \\
 				  &= \left( \eta (-1)^{b }  \right)^{N} a[0],
 \end{split}
\end{equation}
where again we can write as $\eta^N =(\chi  e^{j\pi b/N})^N$. This expression of $a[0]$ 
 perpetuates that $ \left( \eta (-1)^{b}  \right)^{N} =1$, i.e., $ \eta (-1)^{b} = e^{j\frac{2\pi m}{N}}$, for some $m$, where  $0\leq m<N$.
 We can write $\eta  = (-1)^{b} e^{j 2 \pi \frac{m}{N}}$.
 We can substitute in \eqref{eq:s_p} to get 
 \begin{equation}
a[p]=(-1)^{bp} e^{j 2 \pi \frac{m p}{N}} e^{-j \pi \frac{b p^{2}}{N}} a[0], \quad 0 \leq p<N. \label{eq:sk_step5}
\end{equation}

Now that we have the expression for $a[k]$, we can find the autocorrelation---as follows

\begin{align}
\Gamma_{v} & =\sum_{k=0}^{N-1} a[k]^{\ast} a[k+v]\\
& 	\begin{aligned} =\sum_{k=0}^{N-1} &\left(e^{j\pi bk} e^{j 2 \pi \frac{m k}{N}} e^{-j \pi \frac{b k^{2}}{N}} a[0]\right)^{\ast}\\ &e^{j\pi b(k+v)} e^{j 2 \pi \frac{m (k+v)}{N}} e^{-j \pi \frac{b (k+v)^{2}}{N}} a[0] \end{aligned} \\
& =e^{j\pi bv} e^{j 2 \pi \frac{m v}{N}} e^{-j \pi \frac{b v^2}{N}}  \lvert a[0]\rvert^{2} \left(\sum_{k=0}^{N-1} e^{-j 2\pi \frac{b kv}{N}}\right). \label{eq:autocorr}
\end{align}

The autocorrelation for $v = 0$ is
\begin{equation}
\Gamma_{0} = N\lvert a[0]\rvert^{2} = N,
\end{equation}
which implies that $a[0]$ is a pure phase. That is, the modulus of $a[0]$ is unity. This value can be chosen arbitrarily without any loss of generality. In addition, to make the autocorrelation zero for $0<v<N$, $bv$ must be non-divisible by $N$. Thus, it is necessary and sufficient that $b$ be coprime with $N$.

 Since $\gcd(b,N) = 1$, there exists an integer $q$ such that $b q=m\Mod{N}$, which can help us simplify the expression of the searched sequence in  \eqref{eq:sk_step5} to

\begin{equation}
a[k] = e^{-j \pi bk \frac{k-N- 2q}{N}} a[0]. \label{eq:sk_step6}
\end{equation}

Let  $c=N\Mod 2$ be a parameter that characterizes the parity of $N$. Then, the sum $-2 q-N$, which has the parity of $N$, can be written as $c+2 m$, for some integer $m$. We can rewrite \eqref{eq:sk_step6} as

\begin{equation}
a[k]=e^{-j \pi b k \frac{k+c+2 m}{N}} a[0]. 
\end{equation}

Without loss of generality, we can take $a[0]=1$ (note that we can take any complex value with unit modulus). Then, we have

\begin{equation}\label{eq:endProofTh1}
a[k]=e^{-j \pi b k \frac{k+c+2 m}{N}}. 
\end{equation}
\end{proof}

\noindent Indeed, the expression in \eqref{eq:endProofTh1} is the general expression of a ZC sequence, and by taking $m = 0$, we have $a[k] = z_b[k]$ in \eqref{eq:defn}. In conclusion, Zadoff-Chu sequences are the only sequences that convert a shift in frequency into a shift in time, as seen in the relationship in \eqref{ZC_mod_time}.

We can follow a similar approach to the proof of Theorem \ref{Th:freqShift}, which we include in Appendix \ref{apen:timetoFreq}. Otherwise, since the DFT of ZC sequences is a ZC sequence, then---by taking the DFT of \eqref{ZC_mod_time}---we get a similar relationship to   \eqref{ZC_mod_freq}.

\begin{theorem}\label{Th:time2Freq}
The ZC  sequence is the only family of CAZAC sequences with the relationship in  \eqref{ZC_mod_freq}.
 
\end{theorem}

\begin{proof}
See Appendix \ref{apen:timetoFreq}.
\end{proof}
Now that we have demonstrated the unique properties of the ZC sequences, we describe their performance in doubly selective channels. 

\subsection{ZC-modulation for LTV systems}\label{sec:designR}
The ZC sequences,  \eqref{eq:defn}, have two parameters, $N$ and $R$. 
If we take $R=1$, the ZC sequence boils down to the OCDM modulation used in Section \ref{sec:ocdm}. Thus, properly designing the parameter $R$ according to the channel memory would better mitigate the unresolvable propagation paths. In what follows, we describe the design of the parameters and then characterize the effective channel of the ZC-modulated signals. 
\subsubsection{Designing the ZC parameters}
There are two ways to design the ZC sequences, depending on whether we wish to transform the doubly dispersive channel into an equivalent time-dispersive channel by using the relationship in \eqref{ZC_mod_time}, or into a frequency-dispersive channel by exploiting the relationship in \eqref{ZC_mod_freq}.

If we wish to transform the doubly selective channel into a channel with only time-dispersion, we translate each frequency shift that is a multiple of $1/N$ into a time shift by $P$ units, where $P$ is the maximum delay spread. Therefore, we design $R$ that satisfies \eqref{ZC_mod_time}. We can rewrite \eqref{ZC_mod_time} as 
\begin{equation} \label{eq:designR_timeEven}
  e^{-j\pi R \frac{k^2}{N} } e^{\frac{j2\pi k}{N}} =  e^{\Psi}  e^{ -j\pi R \frac{(k-P)^2}{N} }, 
 \end{equation}
when $N$ is even, and 
\begin{equation}  \label{eq:designR_timeOdd}
  e^{-j\pi R \frac{k(k+1)}{N} } e^{\frac{j2\pi k}{N}} =  e^{\Psi}  e^{ -j\pi R \frac{(k-P)(k-P+1)}{N} }, 
 \end{equation}
 when $N$ is odd. By solving either \eqref{eq:designR_timeEven} or \eqref{eq:designR_timeOdd} for $R$, we find that 
\begin{equation}\label{eq:valueR_timeEven}
P = R^{-1} \Mod{N}.
 \end{equation}
Here, $R^{-1}$ denotes the modular multiplicative inverse of $R$ modulo $N$, i.e., the unique integer $P \in \{1, 2, \ldots, N-1\}$  such that $RP \equiv 1 \pmod{N}$. This inverse exists because $R$ and $N$ are coprime by construction. 
The design rule in \eqref{eq:valueR_timeEven} simply says that given the maximum delay index $l_{\max}$ and Doppler index $\kappa_{\max}$, select the smallest $P$ coprime with $N$ that satisfies $P>l_{\max}$ and $2P\kappa_{\max}+l_{\max}<N$, then obtain $P = R^{-1} \Mod{N}$ by the extended Euclidean algorithm at a cost of $\mathcal{O}\left( \log^2(N)\right)$ operations. The first condition prevents paths with distinct delays from colliding, the second prevents the projection from wrapping, and the two are jointly feasible whenever $(l_{\max}+1)(2P\kappa_{\max}+1)\leq N$. 

Now, we can deduce the phase of the multiplicative factor to be $\Psi = \frac{\pi RP^2}{N}$ for even values of $N$, and $\Psi = \frac{\pi RP(P-1)}{N}$ for odd values of $N$. On the other hand, if we wish to work with a frequency-dispersive channel, we transform each unit delay shift into a frequency shift by $\frac PN$. Hence, we can express \eqref{ZC_mod_freq} as 

\begin{equation} \label{eq:designR_freqEven}
  e^{-j\pi R \frac{(k-1)^2}{N} } =  e^{\Psi}  e^{ -j\pi R \frac{k^2}{N} } e^{\frac{j2\pi kP}{N}}, 
 \end{equation}
when $N$ is even, and 
\begin{equation}  \label{eq:designR_freqOdd}
  e^{-j\pi R \frac{k(k-1)}{N} }  =  e^{\Psi}  e^{ -j\pi R \frac{k(k+1)}{N} }e^{\frac{j2\pi kP}{N}}, 
 \end{equation}
 when $N$ is odd. Solving \eqref{eq:designR_freqEven} and \eqref{eq:designR_freqOdd} yields 
 \begin{equation}
P = R \Mod{N}, 
 \end{equation}
 and we find that the phase of the multiplicative factor is $\Psi = \frac{-\pi R}{N}$ for even values of $N$, and $\Psi = 0$ for odd values of $N$.

We study the effective channel after we have translated the doubly-dispersive channels onto a time-dispersive equivalent channel using the ZC sequences, as outlined in  \eqref{ZC_mod_time}.

\subsubsection{ZC Channel as a Convolutional Filter}
In what follows, we focus on the first scenario, where we translate any Doppler shift into a time shift of $P$ units. We show that we can express the received signal $\vect{r}$ as a convolutional filter of the transmitted signal $\vect{x}$ and the channel in the chirp domain $\matr{H}_{\text{ZC}}$, which we define shortly. 
Unlike other waveforms designed for doubly dispersive channels, such as OTFS, viewing $\matr{H}_{\text{ZC}}$ as a convolutional filter allows the construction of a trellis structure for the communication system. 
As such, we can implement a maximum likelihood receiver using the efficient implementation of the Viterbi algorithm, or soft-output receivers using the BCJR, log-MAP, max-log-MAP, or SOVA algorithms. 
The receiver, therefore, detects symbol by symbol as the frame arrives, with latency set by the channel memory rather than by the frame length as done in OTFS systems. Once the receiver projects the frame onto the ZC basis, it detects symbols sequentially along the chirp-domain index and releases decisions after a traceback delay set by the channel memory rather than after a block detector converges over the entire frame.

Using a similar approach to Section \ref{sec:ocdm}, a signal $\vect{x}$ modulated by the cyclic-inverse of a ZC sequence of parameters $R = P^{-1}\Mod{N}$ and length $N$ can be expressed as 
\begin{equation}
\vect{s} = \matr{C}_R^H\vect{x},
\end{equation}
where 
\begin{equation}\label{eq:zcmatrix}
C_R[m,n] = \frac{1}{\sqrt{N}} z_R[m-n]^*= \frac{1}{\sqrt{N}} e^{j \frac{\pi R (m-n)^2}{N}}, 
\end{equation}
and we normalize by $\sqrt(N)$ to keep the energy per transmitted symbol $x[k]$ equal to unity.
Then, we can write the input-output relationship of a ZC-modulated communication system passing through the doubly selective channel in \eqref{eq:channelMatrix} as 
\begin{equation}\label{eq:noNoise_IO_zc}
\vect{y}_{\text{ZC}}  = \matr{C}_R \vect{r}  = \matr{C}_R \matr{H}  \matr{C}_R^H \vect{x} = \matr{H}_{\text{ZC}} \vect{x},
\end{equation}
where $ \matr{H}_{\text{ZC}}$ is the channel matrix modulated by the ZC sequence $\vect{z}_R$ of length $N$. 
We can expand the expression of $ \matr{H}_{\text{ZC}}$ to
\begin{equation}
\begin{split}
{H}_{\text{ZC}} \left[m,n\right] 	&= \sum_{i_1=0}^{N-1}\sum_{i_2=0}^{N-1} {C}[m, i_1]  H[i_1, i_2] {C}^{\ast}[i_2,n]\\
						&\begin{aligned} =\frac 1N e^{\frac{j\pi R(m^2-n^2)}{N}} &\sum_{i_1=0}^{N-1}\sum_{i_2=0}^{N-1}  \sum_{i=0}^{L-1} h_i \delta[i_1-i_2-l_i] \\& e^{j \frac{\pi R }{N} \left(  i_1^2 -i_2^2 + 2i_2n - 2i_1m\right)} e^{j \frac{2 \pi i_2 \kappa_i}{N} } \end{aligned} 
\end{split}\end{equation}
By following the same steps as in Appendix \ref{apen:chirp_channel}, where we take $i_2 = i_1-l_i$, we obtain 
\begin{equation}
\begin{split} {H}_{\text{ZC}} \left[m,n\right] 	=&
\frac 1N e^{\frac{j\pi R(m^2-n^2)}{N}} \sum_{i=0}^{L-1} h_i  e^{j\frac{\pi}{N} (2l_i\kappa_i -  Rl_i^2 - 2nRl_i )} \\&
						\sum_{i_1=0}^{N-1}  e^{j \frac{2 \pi R i_1}{N}(n - m + l_i )}  e^{j \frac{2 \pi i_1 \kappa_i}{N}} .
                        \end{split}\label{eq:zc_Hzc}
\end{equation}

To simplify the last summation term, we note that 

\begin{equation}\label{eq:Dirichlet}
\begin{split}
\mathcal{D}_i(m, n)&=\sum_{i_1=0}^{N-1}  e^{j \frac{2 \pi i_1}{N} \left( R(n - m + l_i ) +  \kappa_i\right)}\\& =\begin{cases}
						N, & ~\text{if}~ R(n - m + l_i ) +\kappa_i = 0 \Mod{N}\\
						0, & ~ \text{otherwise}.
						\end{cases}
\end{split}  
\end{equation}

\noindent Then, we obtain the following form for  $ \matr{H}_{\text{ZC}}$
\begin{equation}
\begin{split}
  {H}_{\text{ZC}} \left[m,n\right] 	=&  e^{\frac{j\pi R(m^2-n^2)}{N}} \sum_{i=0}^{L-1} h_i  e^{j\frac{\pi}{N} (2l_i\kappa_i -  Rl_i^2 - 2nRl_i )} \\& \delta[ R(n - m + l_i ) +\kappa_i  \Mod{N}].  
\end{split}  \label{eq:setp_b} 
\end{equation}

\noindent To confirm that OCDM is a special case of a ZC-modulated system, take $R=1$ in \eqref{eq:setp_b}, then we have \eqref{eq:ocdm_Ris1}.


\subsubsection{Characterizing the Contributing Paths of the Virtual Channel}\label{sec:ZCchannel}
From  \eqref{eq:setp_b}, we can see that a path, $i$, contributes to the effective channel if and only if $R(n - m + l_i ) +\kappa_i = 0 \Mod{N}$. Then, we can say that the effective channel of a ZC-modulated system is a function of the values of $l_i, \kappa_i$, in relation to $N$ and $R$.
 
To simplify the expression of $\matr{H}_{\text{ZC}}$ in \eqref{eq:setp_b}, we characterize which path $i$ contributes to the channel. 
Equivalently, we have $n - m + l_i  +R^{-1} \kappa_i=0 \Mod{N}$ or $n = m - l_i  -R^{-1} \kappa_i \Mod{N}$. 
Suppose that all paths are such that $l_i \in\{0,1, \ldots, P-1\}$ for some $P$ and $\kappa_i \in\{0,1, \ldots, G-1\}$. Therefore, there exists a maximum of $PG$ paths.

On the one hand, if we do not want the contributions of the $PG$ paths to overlap in the determination of $\matr{H}_{\text{ZC}}$, so that each path is resolvable and separable, we need $n = m - l_i  -R^{-1} \kappa_i \Mod{N}$ to take different values for all $l_i$  and $\kappa_i$. On the other hand, if we want to use a trellis-based maximum likelihood equalizer, we need to compact the virtual channel to minimize the complexity. Recall from \eqref{eq:valueR_timeEven}, we designed $R = P^{-1}$, and since $N$ is even, and often taken as a pure power of 2, this means that $P$ should always be taken as odd in order to be coprime with $R$. 
If the above conditions are satisfied, then \eqref{eq:setp_b} can be simplified to 
\begin{equation}\label{eq:zc_channel}
\begin{split}
{H}_{\text{ZC}} \left[m,n\right] = \sum_{i=0}^{L-1} h_i  e^{-j \frac{\pi}{N}\kappa_i^2P} e^{-j \frac{2\pi}{N} l_i \kappa_i} e^{j \frac{2\pi}{N} m \kappa_i} \\ \delta[n - m + l_i  +P \kappa_i\Mod{N}]. \end{split}
\end{equation}

 \noindent This formula is also valid for any choice of $P$, smaller than the maximum delay spread but with some possible overlapping, meaning that more than one path may contribute to a tap of the virtual channel. The design therefore degrades gracefully under support mismatch. Underestimating $l_{\max}$ shrinks $P$ and reintroduces collisions, recovering OCDM behavior in the limit $P=1$; overestimating it preserves path separability but lengthens the effective channel memory. We accordingly bias $P$ upward whenever the delay spread is uncertain, and since the channel memory may be large, we adopt suboptimal but affordable-complexity algorithms, like LMMSE, whose complexity does not grow with the memory. 
 The virtual channel of a ZC-modulated system in \eqref{eq:zc_channel} shows that the contribution of the $i$-th path in the virtual channel has a complex fading of $\tilde{h}_i = h_ie^{-j \frac{\pi}{N}\kappa_i^2P} e^{-j \frac{2\pi}{N} l_i \kappa_i}$, a time-shift of $l_i  +R^{-1} \kappa_i \Mod{N}$, and varies in time with a pure frequency of $\frac{\nu_i}{N}$, which is identical to the Doppler shift in frequency. 
Hence, if there are two or more contributions to the $i$-th tap ${H}_{\text{ZC}} \left[m,n\right] $, we can have a superposition of two or more sinusoidal signals with different frequencies.


\paragraph{Fractional Doppler Considerations}
The effective channel in \eqref{eq:zc_channel} represents the assumption that the delay-Doppler shifts align with the sampling grid. We adopt this on-grid assumption for analytical tractability, following the convention in the OTFS and AFDM literature \cite{2023_Benami, Ravi2018_MPA, Ravi2019_EPA}, which is a fair approximation when the delay and Doppler resolutions are sufficiently small that real values can be approximated to the nearest integer. The assumption is nonetheless optimistic for practical Doppler values, which motivates the fractional-Doppler analysis that follows. 
Hence, we study the ZC-based modulation in the presence of fractional Doppler values. 
That is, we assume that the Doppler shift can be expressed as $\kappa_i = \frac{\nu_i + v_i  }{T_s}$, where $v_i \in (-\frac{1}{2}, \frac{1}{2}]$ is the fractional residue. Having a fractional Doppler affects the expression of the channel matrix, specifically \eqref{eq:Dirichlet}, which is only true for integer Doppler values. Moreover, the extent to which fractional Doppler degrades the effective channel structure depends strongly on the pulse-shaping filter.
We present the fractional delay-Doppler study for the rectangular filter in the frequency domain, corresponding to a sinc function in the time domain and to a square-root raised-cosine filter with roll-off factor $\beta = 0$, which is a rather pessimistic scenario. 

After matched filtering, the effective discrete-time channel incorporates pulse shaping by sampling the combined Nyquist pulse at Doppler-shifted instants. We  can generalize \eqref{eq:Dirichlet} as 
\begin{equation} \label{eq:dirichlet_kernel}
\begin{split}
\mathcal{D}_i(m, n)\big|_{\beta=0}  &= \sum_{i_1=0}^{N-1} e^{j\frac{2\pi i_1}{N}\left(R(n - m + l_i) + \nu_i + v_i \right)} \\
&= \frac{\sin\!\left(\pi\left(R(n-m+l_i) + \nu_i + v_i \right)\right)}{\sin\!\left(\frac{\pi}{N}\left(R(n-m+l_i) + \nu_i + v_i \right)\right)}, 
 \end{split}\end{equation}
whose magnitude peaks when $R(n-m+l_i)+\nu_i + v_i \approx0 \Mod{N}$ and decays as the argument moves away from this point, with the rate of decay increasing with $N$. This slow decay is attributed to the choice of $\beta = 0$, and any practical roll-off $\beta > 0$ yields faster decay and hence a more compact effective channel. 
Consequently, each physical propagation path no longer maps to a single tap in the chirp domain, but instead contributes a spread of non-zero entries centered at the nominal integer position $n=m-l_i-P\nu_i \Mod{N}$. 

 The effective channel \eqref{eq:zc_Hzc} under fractional Doppler can then be written as 
\begin{equation} 
\begin{split}
&H_{\text{ZC}}^{\text{frac}}[m,n] \\&\approx e^{j\frac{\pi R(m^2 - n^2)}{N}} \sum_{i=0}^{L-1} h_i \, e^{j\frac{\pi}{N}(2l_i\kappa_i - Rl_i^2 - 2nRl_i)}  \frac{\mathcal{D}_i(m,n)\big|_{\beta=0} }{N}. \label{eq:Hzc_frac} 
\end{split}  \end{equation}

Following a similar methodolgy to  \cite{Ravi2018_MPA} and \cite{2023_Benami}, we adopt a truncation threshold $k_v$ chosen such that $|\mathcal{D}_i (m,n)|/N$ falls below a sensitivity floor for $|n-(m-l_i-P \nu_i)|>k_v$, as such each path effectively occupies $2k_v+1$ virtual taps. Thus, we can approximate the effective channel as 
\begin{equation} 
\begin{split}
&H_{\text{ZC}}^{\text{frac}}[m,n] \approx\sum_{i=0}^{L-1}\\&  \sum_{q=-k_\nu}^{k_\nu} \tilde{h}_i^{(q)} e^{j \frac{2 \pi m\left(\nu_i+q\right)}{N}} \delta\left[n-m+l_i+P\left(\nu_i+q\right) \Mod{N}\right]
\end{split}  \end{equation}
where $\tilde{h}_i^{(q)}$ is path-dependent weighting coefficients capturing the leakage of the fractional Doppler. 
This behavior is structurally identical to that of AFDM and OTFS under fractional Doppler \cite{2023_Benami}, where the path location similarly widens from a single index to an interval of width $2k_v+1$. In fact, none of the three schemes is immune to this spreading, and the degree of spreading depends on the magnitude of $v_i$, the frame length $N$, and the filter. The primary implication for the ZC-modulated system is that the trellis-based MLSE receiver, which benefits from the exact convolutional form of \eqref{eq:setp_b}, becomes approximate under fractional Doppler as the effective channel memory grows.

Overall, we select the ZC slope/root parameter to minimize the overlap of the propagation paths in the virtual domain, at the cost of increasing the channel memory.


\subsection{Relationship to Affine Frequency Division Multiplexing}
As a DAFT-based chirp waveform, AFDM, like ZC-based modulations, generalizes OCDM through tunable chirp parameters and circumvents chirp-domain fading. In what follows, we discuss the relationship between ZC-based modulation and AFDM\footnote{We limit the comparison to AFDM in the paper, and we refer interested readers to \cite{deng2025_survey, rou2024_mag} to study the comparison between OTFS (and delay-Doppler waveforms) and chirp communication.}. 

AFDM modulates information symbols using the inverse DAFT (IDAFT). A DAFT matrix of a length $N$ can be expressed as a composition of three operations, given by
 \begin{equation}\label{eq:afdm}
\mathbf{A}=\matr{\Lambda}_{c_2} \matr{F} \matr{\Lambda}_{c_1}.
\end{equation}
where $\matr{\Lambda}_{c_1}=\operatorname{diag}\left(e^{-j 2 \pi c_1 k^2}\right)$ is the pre-chirp diagonal multiplication, and  $\matr{\Lambda}_{c_2}=\operatorname{diag}\left(e^{-j 2 \pi c_2 n^2}\right)$ is the post-chirp diagonal multiplication, where $c_1$ and $c_2$ are continuous-valued parameters. 
Similarly, the ZC modulation matrix $\mathbf{C}_R$ in \eqref{eq:zcmatrix} can be decomposed by expanding the quadratic phase of the ZC sequence. Specifically, we can write
\begin{equation}\label{eq:zcmatrix2}
\mathbf{C}_R=\matr{\Lambda}_R \matr{F}_R \matr{\Lambda}_R,
\end{equation}
where $\matr{\Lambda}_R=\operatorname{diag}\left(e^{j \pi R k^2 / N}\right)$ and $\matr{F}_R$ is a matrix with entries $(1 / \sqrt{N}) e^{-j 2 \pi R m n / N}$. Since $R$ is coprime with $N$, the matrix $\matr{F}_R$ is a column-permuted version of the standard DFT matrix. 
As we compare \eqref{eq:afdm} and \eqref{eq:zcmatrix2}, we observe that the DAFT uses two independent continuous parameters, $c_1$ and $c_2$, for the pre-chirp and post-chirp stages. In contrast, the ZC modulation couples both chirp stages and the frequency mapping through a single integer parameter, $R$. This coupling is a consequence of the algebraic structure of ZC sequences, which, as shown in Theorems \ref{Th:freqShift} and  \ref{Th:time2Freq}, is necessary and sufficient for a CAZAC sequence to convert a two-dimensional channel dispersion into a one-dimensional dispersion.

Moreover, the columns of $\mathbf{C}_R^H$ are cyclically shifted ZC sequences, which, by definition, are guaranteed to be CAZAC sequences for any admissible value of $R$. 
In contrast, the columns of the IDAFT matrix are chirp sequences of the form $(1 / \sqrt{N}) e^{j 2 \pi\left(c_1 k^2+k n / N+c_2 n^2\right)}$, which satisfy the constant-amplitude condition by construction since the basis functions are complex exponentials with unit modulus; however, the zero autocorrelation property depends on the specific choice of $c_1$ and $c_2$ and is not guaranteed unless the parameter values are chosen appropriately.
The CAZAC property of the ZC basis is practically significant because it ensures ideal periodic autocorrelation, regardless of the choice of $R$.
This feature makes ZC sequences the standard preamble and reference-signal family in deployed systems, and it suggests that a ZC-modulated frame could reuse one sequence for synchronization, channel estimation, and data transport. Establishing this requires dedicated pilot and estimator design, which we leave to future work. 

Both ZC-based modulation and AFDM address the chirp-domain fading problem of OCDM by designing their respective parameters to separate overlapping delay-Doppler paths. In AFDM, the parameters $c_1$ and $c_2$ are selected by setting $c_1$ proportional to the inverse of the maximum Doppler index scaled by $N$, such that all paths in the DAFT domain map to distinct delay indices. The design rule operates over continuous values, which offers fine-grained control over path placement.
In ZC-based modulation, the design operates through the integer root parameter $R$, which satisfies \eqref{eq:valueR_timeEven}. The resulting effective channel in \eqref{eq:zc_channel} maps the $i$-th propagation path to a delay of $l_i+P \kappa_i(\bmod N)$ in the chirp domain. The mapping, $\left(l_i, \kappa_i\right) \mapsto l_i+P \kappa_i(\bmod N)$, is a linear projection from the 2D delay-Doppler plane onto a 1D axis, where $P$ controls the slope of this projection. This projection makes it convenient to directly inspect where each path falls on this axis for a given channel profile. The key mechanism is different from that of AFDM, where ZC modulation exploits the algebraic time-frequency duality of ZC sequences (cf. \eqref{ZC_mod_time} and \eqref{ZC_mod_freq}) to transform each Doppler shift into an additional time shift of $P$ units per Doppler index.

The comparison between ZC-based modulations and AFDM can also be understood geometrically. The operations composing the DAFT, chirp multiplications, and the Fourier transform act on the delay-Doppler plane as shears and rotations, respectively. Their composition is an area-preserving transformation of the two-dimensional spreading function, followed by projection onto the one-dimensional DAFT-domain axis. However, an area-preserving transformation of the plane composed with a projection along a fixed axis is equivalent to a projection along a different axis \cite{LCT_book}. That is, any such pre-transformation of the delay-Doppler point cloud can be absorbed into a change of the projection direction.

A key practical consequence of the ZC time-frequency duality is the structure of the effective channel $\mathbf{H}_{\text {ZC }}$ in \eqref{eq:zc_channel} as a convolutional filter with a trellis representation. 
This trellis structure enables the direct application of classical maximum likelihood sequence detection methods, such as the Viterbi algorithm and MAP symbol detection (such as the BCJR and the log-MAP algorithms and their max-log-MAP approximation and the SOVA algorithm), which enable a precise computation of soft information (i.e.,  log-likelihood ratios (LLRs), extrinsic and intrinsic information). In practice, no communication system operates without error correction coding, which requires soft information, and the trellis structure makes turbo processing particularly efficient. The trellis-compatible convolutional form of the ZC effective channel is a distinctive feature that arises directly from the uniqueness property established in Theorems \ref{Th:freqShift} and  \ref{Th:time2Freq}.
In AFDM, the effective channel in the DAFT domain also exhibits a sparse banded structure when $c_1$ and $c_2$ are properly chosen, and equalization can be performed using message passing or linear receivers.  This reduction from two dimensions to one dimension allows a wireless communication designer to reuse mature receiver hardware and software for doubly dispersive channels.

Finally, we remark that ZC-based modulations and AFDM are best viewed as complementary approaches to chirp-domain communication rather than competing approaches. AFDM provides a flexible, continuously parameterized framework with provable full-diversity guarantees, while ZC-based modulation characterizes the unique algebraic structure within the CAZAC family that enables transforming doubly dispersive channels into singly dispersive channels. The ZC framework thus offers theoretical insight into why chirp-domain modulation is effective, and AFDM provides a practical design methodology for optimizing chirp parameters.


\subsection{Diversity Analysis} \label{sec:diversity}

In this section, we analyze the diversity order of a ZC-modulated communication system. In the previous sections, we focused on the input-output relationship of a ZC-modulated system in a noise-less channel. We now extend the input-output relationship in \eqref{eq:noNoise_IO_zc} to an additive white Gaussian noise  (AWGN) channel, $\vect{n} \in \mathbb{C}^N$, 
of zero mean and variance of $N_0$. Then, we can express the input-output relationship in an AWGN channel as

\begin{equation}\label{eq:awgnIO_zc}
\vect{y}_{\text{ZC}} =  \matr{C}_R\matr{H}  \matr{C}_R^H \vect{x} + \matr{C}_R\vect{n} =  \matr{H}_{\text{ZC}} \vect{x} + \vect{w}. 
\end{equation}

\noindent Note that $\vect{w}=  \matr{C}_R\vect{n} \sim \mathcal{N}\left(0, N_0\matr{I}\right)$, where $\matr{I}$ is the identity matrix. We can rewrite \eqref{eq:awgnIO_zc} as
\begin{equation} 
\vect{y}_{\text{ZC}} =  \matr{\Psi} \left( \vect{x}\right) \tilde{\vect{h}} + \vect{w}, 
\end{equation}
where $ \matr{\Psi} \left( \vect{x}\right)$ is a concatenation matrix, , which is  defined as 
\begin{equation}
 \matr{\Psi} \left( \vect{x}\right) = \left[ \hat{\matr{H}}_0 \vect{x}, \hat{\matr{H}}_1 \vect{x}, \ldots, \hat{\matr{H}}_{L-1} \vect{x}               \right]_{N\times L}, 
\end{equation}
where $\hat{{H}}_i [m,n] = e^{j \frac{2\pi}{N} m \kappa_i} \delta[n - m + l_i  +R^{-1} \kappa_i \Mod{N}]$, and $\tilde{\vect{h}}  = [\tilde{h}_0, \tilde{h}_1, \ldots, \tilde{h}_{L-1}]^T_{L\times 1}$.

The diversity order of a communication system is the slope of its probability of error as a function of the SNR. We find the average probability of error using 
\begin{equation}
P_e = \mathbb{E}_{ \tilde{\vect{h}} } [P_e(\tilde{\vect{h}})] =  \mathbb{E}_{ \tilde{\vect{h}} } \left[\frac{1}{|\mathcal{X}|} \sum_{\vect{x}\in\mathcal{X}}P_e(\vect{x} \mid  \tilde{\vect{h}} ) \right],
\end{equation}
where $P_e(\tilde{\vect{h}})$ is the probability of error given a channel instance $\tilde{\vect{h}}$, $\mathcal{X}$ is the M-ary finite alphabet of $\vect{x}$, $|\mathcal{X}|$ indicates the cardinality of $\mathcal{X}$, and $\vect{x}$ is normalized such that it has a unit average energy. By applying the union bound, we obtain the following inequality 
\begin{equation}
P_e(\vect{x} \mid  \tilde{\vect{h}} ) \leq  \sum_{\vect{x} \neq \hat{\vect{x}}} \Pr( \vect{x}  \rightarrow  \hat{\vect{x}}\mid \tilde{\vect{h}} ),
\end{equation}
where $\Pr( \vect{x}  \rightarrow  \hat{\vect{x}} )$ is the pairwise error probability (PEP) of confusing $ \vect{x}$ with $\hat{\vect{x}}$ when
$ \vect{x}$ is transmitted. Given maximum likelihood detection at the receiver, we can find the PEP; that is, we have 
\begin{equation}
\label{eq:pep}
\begin{split}
 \Pr( \vect{x}  \rightarrow  \hat{\vect{x}}\mid \tilde{\vect{h}})  &= \Pr(||\vect{y}_{\text{ZC}} -   \matr{\Psi} \left( \hat{\vect{x}}\right) \tilde{\vect{h}} ||^2 \leq ||\vect{y}_{\text{ZC}} -   \matr{\Psi} \left( \vect{x}\right) \tilde{\vect{h}} ||^2) \\&= Q\left( \frac{||\matr{\Psi} \left( \hat{\vect{x}}\right)  \tilde{\vect{h}} - \matr{\Psi} \left( \vect{x}\right)  \tilde{\vect{h}}  ||}{\sqrt{2N_0}}  \right).\end{split}
\end{equation}

For simplicity, let $\boldsymbol\zeta =  \hat{\vect{x}} - \vect{x}, \boldsymbol\zeta \in \mathbb{Z}^{N\times 1}[j]$, where $ \mathbb{Z}[j]$ is the set of complex integers. Then, we have 
\begin{equation}
 \Pr( \vect{x}  \rightarrow  \hat{\vect{x}}\mid \tilde{\vect{h}} ) = Q\left( \frac{||\matr{\Psi} \left( \boldsymbol\zeta \right)   \tilde{\vect{h}}  ||}{\sqrt{2N_0}}  \right), 
\end{equation}
where we can find 
\begin{equation}
||\matr{\Psi} \left( \boldsymbol\zeta \right) \tilde{\vect{h}} ||^2 = \tilde{\vect{h}}^{H}\matr{\Psi} \left( \boldsymbol\zeta \right) ^{H}\matr{\Psi} \left( \boldsymbol\zeta \right)  \tilde{\vect{h}} =\tilde{\vect{h}}^H \matr{\Upsilon} \left( \boldsymbol\zeta \right)  \tilde{\vect{h}},
\end{equation}
where $\mathbf{\Upsilon}(\boldsymbol{\zeta}) = \Psi(\boldsymbol{\zeta})^H\Psi(\boldsymbol{\zeta})$ is an $L \times L$ Hermitian matrix.
Using a simple bound on the Q-function, we find
\begin{equation} 
\begin{split}
P_e 
\leq  \frac{1}{|\mathcal{X}|} \sum_{\vect{x}\in\mathcal{X}} \sum_{\vect{x} \neq \hat{\vect{x}}}  \mathbb{E}_{\tilde{h}} \left[ e^{-\frac{||\matr{\Psi} \left( \boldsymbol\zeta \right) \tilde{\vect{h}} ||^2}{4N_0}} \right]. \end{split}
\end{equation}
Using Chernoff's bound for the Rayleigh fading channel, where $\tilde{\mathbf{h}} \sim \mathcal{CN}(\mathbf{0}, \frac{1}{L}\mathbf{I})$, we can bound the probability of error as follows
\begin{equation} 
P_e  \leq  \frac{1}{|\mathcal{X}|} \sum_{\vect{x}\in\mathcal{X}} \sum_{\vect{x} \neq \hat{\vect{x}}} \left( \prod_{i=1}^{r(\vect{x}, \hat{\vect{x}})} \frac{1}{1+ \frac{\lambda_i^2}{4LN_0}}\right), 
\end{equation}
where $r(\vect{x}, \hat{\vect{x}})$ is the rank of $\matr{\Upsilon} \left( \boldsymbol\zeta \right)$, and $\lambda_i$ is the $i$-th non-zero eignvalue of $\matr{\Upsilon} \left( \boldsymbol\zeta \right)$. In a high-SNR regime, the bound on $P_e$ becomes

\begin{equation} 
P_e  \leq  \frac{1}{|\mathcal{X}|} \sum_{\vect{x}\in\mathcal{X}} \sum_{\vect{x} \neq \hat{\vect{x}}}  \left(\prod_{i=1}^{r(\vect{x}, \hat{\vect{x}})}\lambda_i^2\right)^{-1} \left(\frac{1}{4LN_0}\right)^{-r(\vect{x}, \hat{\vect{x}})}. 
\end{equation}

Finally, the system diversity, taken as the slope of $P_e$, is defined as
\begin{equation}
r \overset{\text{def}}= \min_{\vect{x} \neq \hat{\vect{x}}, (\vect{x}, \hat{\vect{x}}) \in \mathcal{X}\times\mathcal{X}} r(\vect{x}, \hat{\vect{x}}). 
\end{equation}

\begin{figure}
    \centering
    \includegraphics[width=\linewidth]{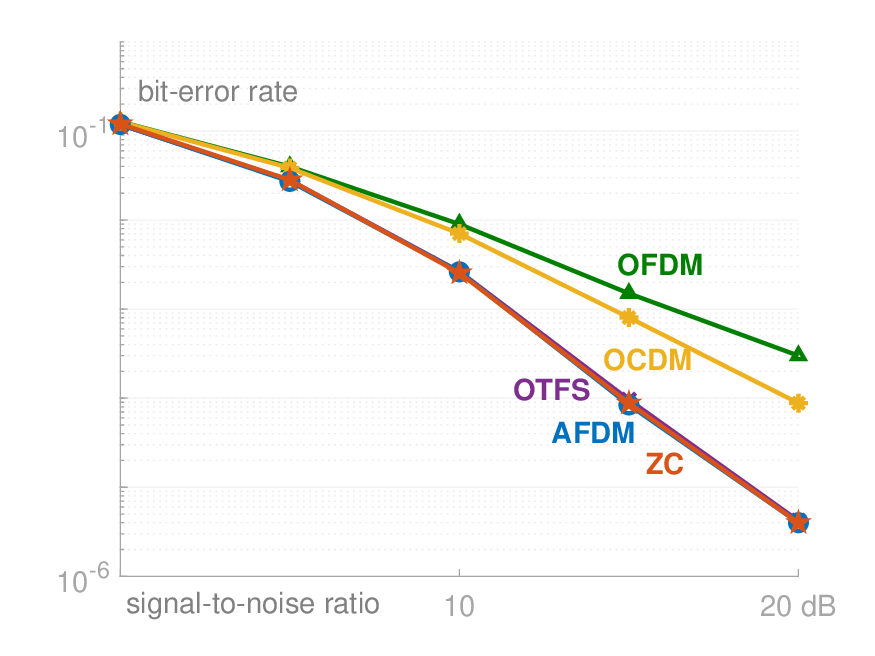}
    \caption{Optimized chirp-domain communication, ZC and AFDM, and OTFS obtain full diversity in an LTV channel with three propagation paths.}\label{fig: diversity}
\end{figure}
To find the rank of  $\matr{\Upsilon} \left( \boldsymbol\zeta \right)$, it is sufficient to find the rank of $\matr{\Psi} \left( \boldsymbol\zeta \right)$ for all non-zero vectors of $\boldsymbol\zeta$. By observing $\matr{\Upsilon} \left( \boldsymbol\zeta \right)$, one can see that it results in a phase shift of the vector $ \boldsymbol\zeta$ and a circulant shift. If  $\boldsymbol\zeta = c \boldsymbol 1_{N\times 1}$, where $1_{N\times 1}$ is all one vector and $c \in \mathbb{Z}[j]$ is a constant, then rank$[\matr{\Upsilon} \left( \boldsymbol\zeta \right)]< L$. Hence, the ZC-modulated systems, including the special case of OCDM, do not ensure full diversity all of the time. Although ZC modulation does not provide a full diversity gain in the asymptotic sense, our numerical results show that ZC can still maintain some degree of diversity in the finite SNR regime.
The diversity limitation identified above is not unique to ZC-based modulation. Surabhi \emph{et al}. \cite{Surabhi2019_diveristy} show that uncoded OTFS modulation exhibits the same structural limitation. AFDM, by contrast, avoids this limitation through its two-parameter continuous design, which provides sufficient degrees of freedom to maintain full rank for all error patterns \cite{2023_Benami}. 
Future work could also investigate classical techniques for extracting the full diversity order, such as phase rotation \cite{damen_diveristy}, or improving system performance by leveraging coding gain.

Another important practical consideration in the design of ZC-based modulations is the relationship between the diversity order and the receiver complexity, both of which depend on the choice of $ P = R^ {-1} \Mod{N}$. From  \eqref{eq:zc_channel}, the effective channel memory is determined by the maximum chirp-domain delay, $\max _i\left(l_i+P \kappa_i \Mod{N}\right)$. As $P$ increases to better separate overlapping paths and thereby increase the achievable diversity, the effective channel memory grows proportionally. For a trellis-based receiver such as the Viterbi or BCJR algorithms, the number of states scales as $|\mathcal{X}|^{L_{\text {eff }-1}}$, where $|\mathcal{X}|$ is the constellation size and $L_{\text {eff }}$ is the effective channel memory length. Consequently, maximizing the diversity order may lead to a prohibitively large trellis, particularly for higher-order modulation alphabets.
The geometric nature of the ZC mapping makes this tradeoff dynamically adjustable. 
Since the delay-Doppler profile of a time-varying channel evolves over a communication session, there is no single fixed $P$. 
The channel profile is projected onto a one-dimensional axis with slope $P$, and one can directly inspect, for each admissible value of $R$, both the number of separated paths (which determines the achievable diversity) and the resulting channel memory (which determines the receiver complexity), allowing $P$ to be updated accordingly as channel conditions change. In this way, ZC modulation exposes a family of operating points on the diversity-complexity curve from which to select the configuration that best matches the target application's requirements.
\newline
This perspective offers a practical advantage. In many scenarios, the full diversity order $L$ is not required to meet a given reliability target. For instance, at a target BER of $10^{-4}$, a diversity order of three may be sufficient to achieve the desired performance, and increasing the diversity beyond this value yields diminishing returns in BER while incurring a significant increase in receiver complexity. 
ZC-based modulation thus provides a principled mechanism for matching the system's diversity order to the operational requirements, rather than unconditionally maximizing diversity at the expense of complexity. We note that the tradeoff between diversity and complexity is less visible in AFDM, where the parameter design is oriented toward guaranteeing full diversity for all channel realizations. While guaranteeing full diversity is a desirable theoretical property, it does not directly expose the diversity-complexity relationship to the system designer.

\begin{figure}
    \centering
    \includegraphics[width=\linewidth]{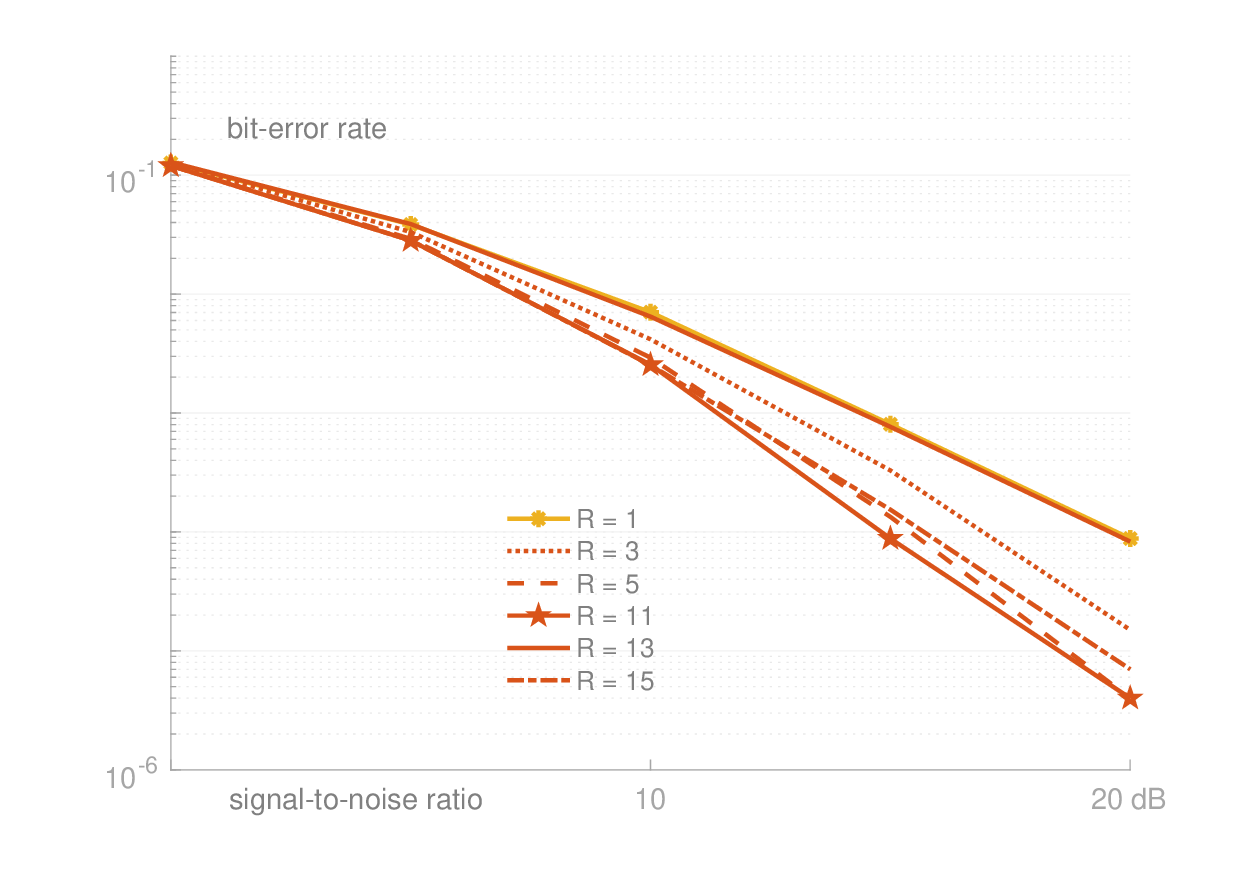}
        \caption{Changing the chirp parameter $R$ yields different error performance for the same channel.}\label{fig:Rvals}
\end{figure}

\section{Results and Discussion} \label{sec:results}

\noindent In this section, we present numerical results to demonstrate the performance of ZC-based modulation in doubly selective channels. We evaluate the bit error performance in a Rayleigh fading channel with Doppler shifts. 

We start by evaluating and comparing the performance of ZC modulation with state-of-the-art modulation schemes, namely OTFS \cite{2019_biglieri}, OCDM, and affine frequency division multiplexing (AFDM) \cite{2021_Benami, 2023_Benami}, the latter of which has been proven to achieve full diversity in doubly selective channels. We consider the parameters suggested in \cite{2023_Benami}, where $L=3$, and the corresponding $R$ value is set to $R=11$, using a binary phase shift keying (BPSK) constellation alphabet and maximum likelihood sequence detection (MLSD). We choose MLSD to better assess the diversity gain of our proposed modulation in LTV channels. For a fair comparison, we use $N=16$ for ZC, AFDM, and OCDM, while the OTFS frame is generated with $N_{OTFS} = 4$ and $M_{OTFS} = 4$, as shown in  \cite{2023_Benami}, with an ideal pulse-shaping window and a cyclic prefix of length equal to the maximum delay spread. All four schemes operate with the same bandwidth, frame duration, constellation alphabet, and cyclic prefix length.
We plot the bit-error-rate (BER) for the four systems in Fig. \ref{fig: diversity}, which shows that the proposed ZC-based modulation system achieves full diversity in the finite SNR regime---similarly to OTFS and AFDM. Therefore, Fig. \ref{fig: diversity} confirms that designing the chirp rate according to the channel spread helps mitigate the chirp selectivity. 
It is worth mentioning that when the propagation paths in the chirp domain are separated, OCDM performs similarly to ZC, OTFS, and AFDM.

In Section \ref{sec:designR}, we show a systematic way of designing the chirp parameter, $R$, to achieve the desired relationship given by \eqref{eq:valueR_timeEven}. 
In  Fig. \ref{fig:Rvals}, we investigate if other values of $R$ can provide the same error performance as   \eqref{eq:valueR_timeEven}. As expected, the trivial case of $R=1$ provides the same performance as OCDM in Fig. \ref{fig: diversity}. Moreover, Fig. \ref{fig:Rvals} confirms that the value of $R$ is not unique, and there exist other values of $R$ that provide similar performance as \eqref{eq:valueR_timeEven}. However, it is essential that we consider the channel profile when selecting $R$. Thus, the design criteria in Section \ref{sec:designR} provides a systematic framework for tuning $R$ based on the channel maximum delay spread, such that we minimize the overlap of the propagation paths in the chirp domain.

    \begin{figure}[t]
    \centering
    \includegraphics[width=\linewidth]{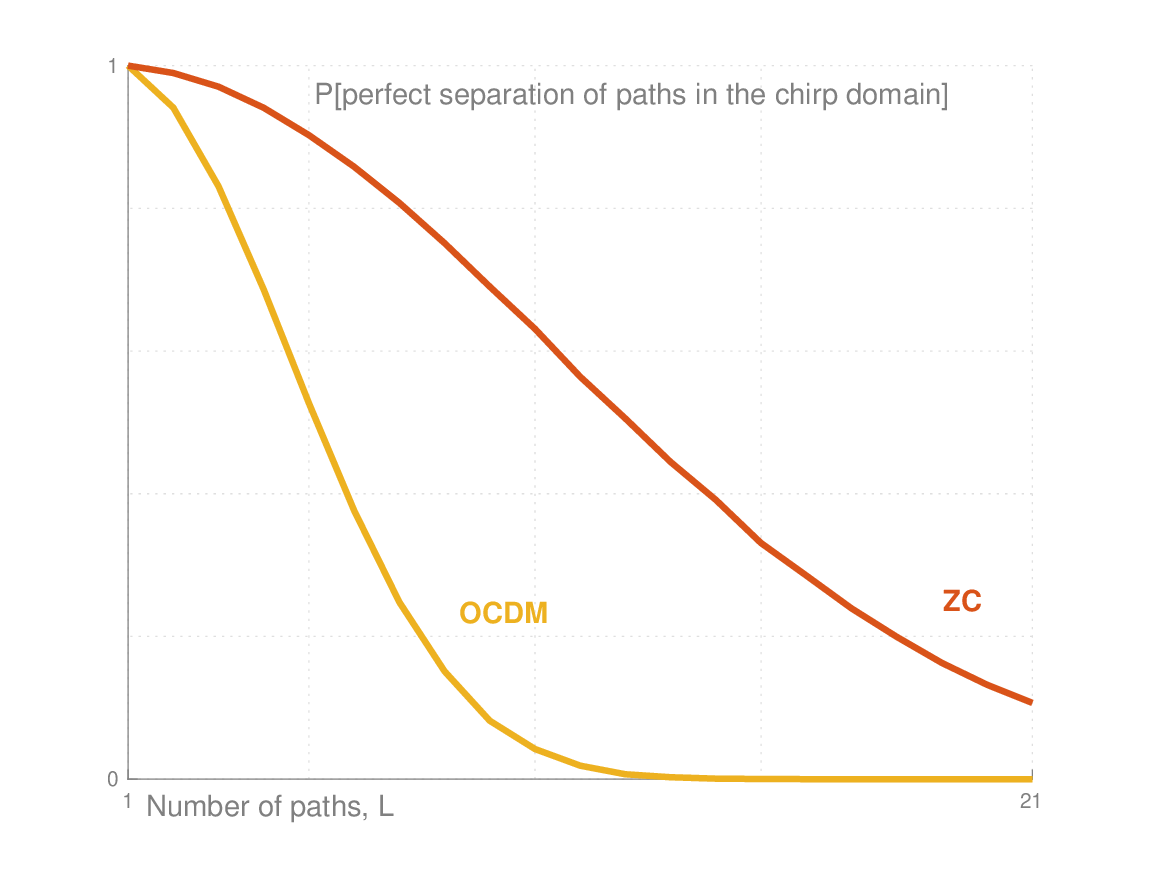}
    \caption{OCDM is more likely to have overlapping paths in the chirp domain, with probability one when $L \geq 11$. }\label{fig:overlapProb}
\end{figure}
\begin{figure}[t]
    \centering
    \includegraphics[width=\linewidth]{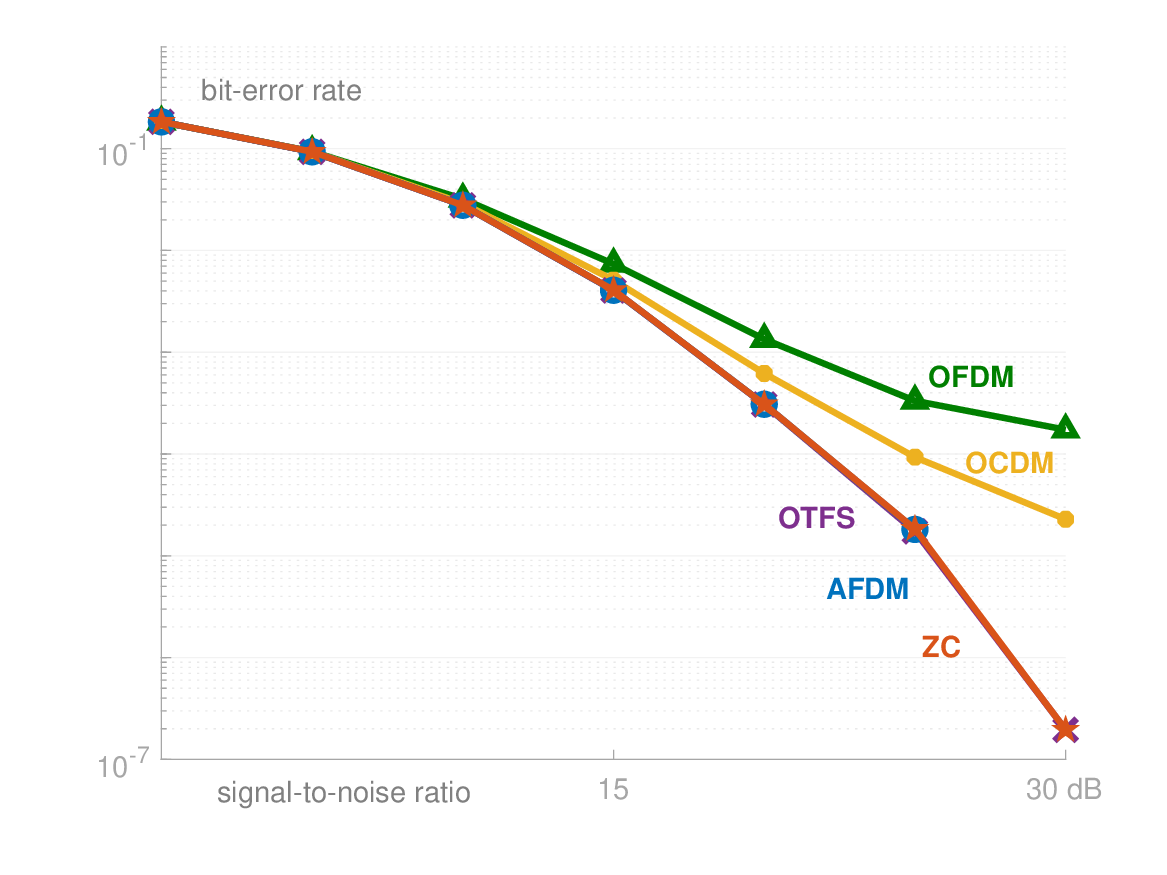}
    \caption{ZC modulation offers 5 dB improvement over OCDM and approximately 8 dB improvement over OFDM using a LMMSE detector in an LTV channel of three paths, assuming integer Doppler shifts. }
    \label{fig:intLMMSE}
\end{figure}

\begin{figure}[t]
    \centering
    \includegraphics[width=\linewidth]{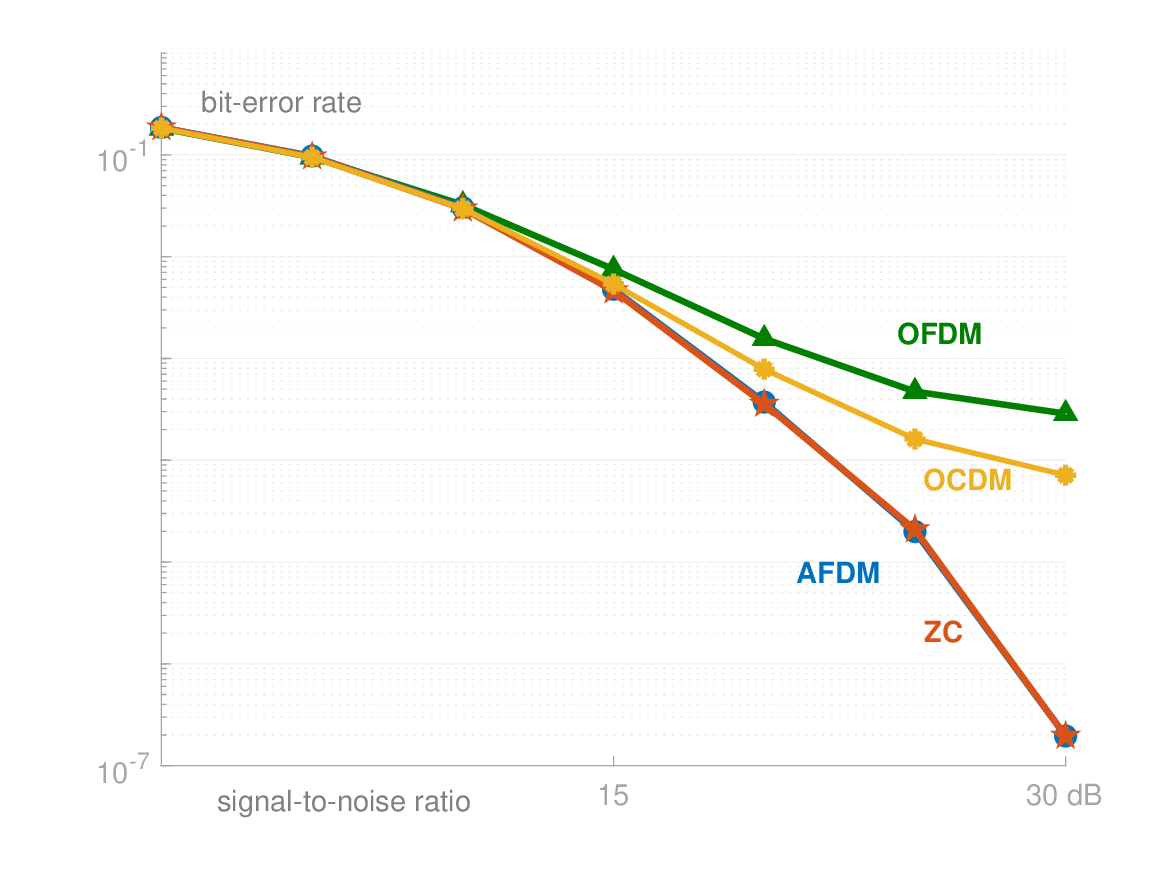}
    \caption{ZC modulation offers reliable communication using a LMMSE detector in doubly selective channels, even with fractional Doppler shifts. }
    \label{fig:fLMMSE}
\end{figure}
    
While OCDM and ZC-based modulation are expected to perform similarly under the same channel conditions, the advantage of using ZC over OCDM is the fact that modulation of the signals using ZC with optimized $R$ is less likely to result in overlapping propagation paths in the chirp domain than OCDM. In other words, by optimizing the parameter of ZC, $R$, one can reduce the event that $l_i  +P \kappa_i\Mod{N} = l_j  +P \kappa_j\Mod{N}$ compared to $l_i +\kappa_i \Mod{N} = l_j +\kappa_j \Mod{N}$, for $i\neq j, i,j = 0,\ldots, L-1$. 
In Fig. \ref{fig:overlapProb}, we show a numerical distribution of the probability of having no overlap of propagation paths in the chirp domain, i.e., all chirp paths are unique, for OCDM and ZC-based modulation systems. The propagation paths are synthetically generated, where the delay is uniformly distributed between zero and nine, and the Doppler paths of each path among $L$ are between $[-5,5]$, for $N = 256$ and $R=57$. 
 Fig. \ref{fig:overlapProb} shows that when a channel has 11 paths or more, OCDM will experience chirp selectivity with probability one, whereas ZC modulation is $50\%$ more likely to experience selectivity (overlapping). 
While neither OCDM nor ZC-based modulation can guarantee the separability of the propagation paths, we can design $R^{-1}$ to avoid all overlaps in the chirp domain; however, this may be at the cost of larger channel memory.

The ZC-based modulation performs similarly to OCDM for LTV channels that do not exhibit selectivity in the chirp domain. One can circumvent the selectivity in the chirp domain by using ZC-based modulation at the cost of increased channel memory. Therefore, to fully take advantage of the proposed ZC-based modulation, it is best to pair it with a detection method that does not rely on channel memory, such as LMMSE.

In Fig. \ref{fig:intLMMSE}, we show the performance of OCDM and ZC-based modulation using an LMMSE receiver. 
We also compare the performance of ZC with state-of-the-art modulation schemes, such as OTFS and AFDM, for $N=256$. We consider a channel consisting of three propagation paths, where the maximum Doppler shift, generated using Jakes' Doppler spectrum, is $2$, corresponding to a $540$ km/h  speed, maximum delay spread of $2$, and $R=85$. 
Specifically, we test the performance of five modulation schemes using LMMSE, three of which are in the chirp domain. The figure shows that AFDM, OTFS, and ZC modulation perform similarly in terms of BER. 
This similarity indicates that for a three-path channel with well-separated delay-Doppler indices, all three schemes achieve comparable path separation. The advantage of ZC modulation in this scenario lies in the simplicity of the design procedure. 
Moreover, the convolutional structure of the ZC effective channel in \eqref{eq:zc_channel} enables the use of trellis-based receivers, which may offer advantages in latency-sensitive applications where block-based detection is less desirable.

On the other hand, OCDM leverages the channel selectivity better than OFDM, yet still cannot benefit from the full diversity of the channel, as some paths may interfere in its chirp domain. We can also study the performance of the proposed chirp modulation in fractional Doppler channels, which is a practical assumption. 
Moreover, we validate the performance of OFDM, OCDM, AFDM, and ZC modulation under a more practical assumption of having fractional Doppler shifts in  Fig. \ref{fig:fLMMSE}. We utilize the parameters suggested in \cite{2023_Benami} for a three-path propagation channel and employ an LMMSE detector. Consistent with previous results, chirp-domain communication offers a reliability gain over OFDM. 
Moreover, Fig. \ref{fig:intLMMSE} and Fig. \ref{fig:fLMMSE} show that, for a channel where the delay-Doppler profile may cause selectivity in the chirp domain, ZC modulation offers an improvement of around $3-8$ dB BER at around $10^{-5}$. 
Therefore, we conclude that ZC-based modulation can exploit some delay-Doppler diversity, which OCDM cannot utilize without coding. Moreover, while ZC-based detection would require a much larger channel memory, using LMMSE is sufficient to equalize the channel at a much lower receiver complexity.

 \begin{figure}[t]
    \centering
    \includegraphics[width=\linewidth]{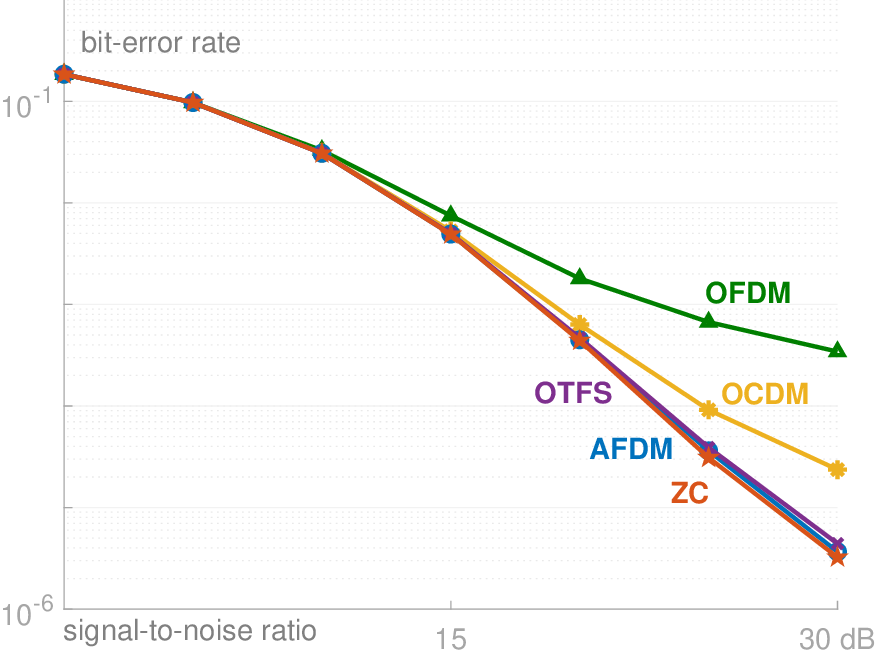}
    \caption{The proposed ZC modulation offers a reliable communication service in the EVA channel model, similarly to AFDM and OTFS. }
    \label{fig:LMMSE_EVA}
\end{figure}

Finally, in Fig. \ref{fig:LMMSE_EVA}, we show the performance of all five modulation schemes in the Extended Vehicular A  (EVA) channel model with a user maximum speed of $500$ km/hr and $N=256$, with subcarrier spacing of $15$~kHz at $4$~GHz carrier frequency. The EVA model consists of nine propagation paths with a maximum delay spread of $2510~ns$, which enables a more realistic evaluation than the synthetic three-path channels. For ZC and AFDM modulations, we set the chirp parameters $R$, $c_1$, and $c_2$ according to the maximum delay and Doppler spread of the EVA profile. 
The performance of the five modulation schemes is as expected and consistent with the previous discussion. 
Specifically, OFDM degrades significantly due to inter-carrier interference (ICI), OCDM outperforms OFDM but suffers from chirp-domain fading caused by overlapping paths, and ZC, AFDM, and OTFS achieve the best performance by resolving the delay-Doppler paths in their respective transform domains. 
Notably, it is evident from Fig. \ref{fig:LMMSE_EVA} that tuning chirp-domain communication can yield up to $10$ dB gain in the high SNR regime. Moreover, our work demonstrates an alternative way of designing chirps using poly-phase sequences, which are commonly used for wireless communication systems.

The numerical results in Figs. \ref{fig: diversity} and \ref{fig:intLMMSE}--\ref{fig:LMMSE_EVA} corroborate the analysis in Section \ref{sec:diversity} and 
 show that ZC-based modulation achieves comparable BER performance to AFDM across a range of practical SNR values and channel conditions, including the EVA channel model. This observation is consistent with the behavior of uncoded OTFS, which similarly does not guarantee full diversity yet performs comparably to full-diversity schemes in the finite SNR regime \cite{Surabhi2020_equalizer}. 
We note that classical techniques, such as phase rotation \cite{damen_diveristy} or outer coding, could be applied to close the asymptotic diversity gap, if required.

\subsection{Complexity Comparison}

To complement the BER comparisons, we present a unified complexity analysis of the modulation and demodulation operations for the five schemes. The modulation/demodulation complexity is $\mathcal{O}(N \log N)$ for OFDM, OCDM, AFDM, and ZC, since each decomposes into element-wise chirp multiplications and FFT operations \cite{rou2024_mag}. OTFS modulation, on the other hand, involves an ISFFT of  $\mathcal{O}(N_{OTFS} \log N_{OTFS})$ followed by a Heisenberg transform of $\mathcal{O}(N_{OTFS} \log M_{OTFS})$, giving an overall complexity of $\mathcal{O}(N_{OTFS}M_{OTFS} \log M_{OTFS})$ \cite{deng2025_survey}.  In terms of detection complexity, LMMSE equalization requires matrix inversion with complexity $\mathcal{O}\left(N ^3\right)$ for OFDM, OCDM, AFDM, and ZC-based modulation and  $\mathcal{O}\left(N_{OTFS} ^3M_{OTFS} ^3\right)$ for OTFS. By taking advantage of the banded and sparse structure of the channel, we can reduce the LMMSE complexity to  $\mathcal{O}\left(N L^2\right)$, where  $L$ is the number of significant ICI terms in OFDM, $L = L_{\text {eff }}^{\text {OCDM }}=\max _i\left(l_i+\kappa_i\right)$ in OCDM, $L = L_{\text {eff }}=\max _i\left(l_i+P \kappa_i(\bmod N)\right)$ in ZC, and finally, $L = L_{\text {eff }}^{\text {AFDM}} = (l_{\max} +1)(2(\nu_{\max} +k_\nu)+1)-1$ in AFDM \cite{bemani2022_lmmse}, and to $\mathcal{O}\left(M_{OTFS} N_{OTFS} \log{M_{OTFS} N_{OTFS} }\right)$ in OTFS \cite{Surabhi2020_equalizer},  
  \cite[and the references within]{deng2025_survey}. With the optimal MLSE receiver, the complexity  is $\mathcal{O}\left(|\mathcal{X}|^{M_{OTFS} N_{OTFS} }\right)$  in OTFS and  $\mathcal{O}\left(|\mathcal{X}|^{N}\right)$ in AFDM. However, for ZC (and OCDM as a special case), we can take advantage of its trellis-compatible structure to use the Viterbi or BCJR algorithms with a trellis complexity of $\mathcal{O}\left(N|\mathcal{X}|^{L_{\text {eff }}}\right)$.

\section{Conclusion} 
\noindent
This article presented a family of chirp-based modulation schemes for time-varying channels (delay-Doppler communication systems). We revealed how chirp signals and ZC sequences efficiently operate on high-mobility wireless channels and highlighted the unique ability of ZC sequences to transform the delay-Doppler channel into a time- or frequency-dispersive channel. We analyzed the interaction of ZC signals with high-mobility channels, extended the analysis to other variants, and compared the proposed approach with the state-of-the-art modulation, such as  OCDM, OTFS, and AFDM. The key advantage of ZC-based modulation, as opposed to its variant OCDM, is that it increases the diversity in the chirp domain and improves multipath diversity gain. 
By tuning the slope parameter of the ZC-based modulation, we provided a geometric projection from the two-dimensional delay-Doppler plane onto a one-dimensional chirp-domain axis to improve the performance of the communication system in the chirp domain. This projection is a function of the chirp root and reveals a trade-off between diversity order and receiver complexity that can be matched to the target application. We showed that ZC-based modulation matches AFDM and OTFS in bit-error rate while reducing chirp-parameter design to a single modular inverse and exposing an explicit diversity-complexity operating trade-off. 
ZC modulation and AFDM frameworks are best understood as complementary. 
The transmitter can also change the root between retransmissions of a packet, at no cost in bandwidth. Each root projects the delay-Doppler plane onto a different axis, so paths that overlap under one root generally separate under another, and a receiver combining both transmissions recovers part of the diversity a single root cannot guarantee. This mechanism needs no channel knowledge at the transmitter, which suits sporadic traffic where $P$ cannot be optimized in advance, and merits a dedicated investigation in future work. 
Finally, the ideal periodic autocorrelation of ZC sequences suggests favorable behavior under channel-estimation error, synchronization offsets, and imperfect knowledge of the delay-Doppler support, and points toward integrated sensing and communication. Quantifying these gains requires pilot and estimator design,  which is a natural continuation of the present work.

\appendices
\section{Deriving OCDM channel matrix  in the time domain}\label{apen:chirp_channel}
\noindent Here, we derive the expression of the channel in the chirp domain, as given in \eqref{eq:chirp_channel}. 
 
By looking at the element of $\matr{H}_c$ in the $m$-th row and $n$-th column, we find 
\begin{equation}
{H}_c \left[m,n\right] = \sum_{i_1=0}^{N-1}\sum_{i_2=0}^{N-1} {\Phi}[m, i_1]  H[i_1, i_2] {\Phi}^{\ast}[i_2,n].
\end{equation}

From the definition of the time-domain channel, $\matr{H}$, in \eqref{eq:channelMatrix}, the $m$-th row and $n$-th column is a summation of the contribution of $L$ propagation paths. For simplicity, we write the expression of one such path without the complex fading. Let $\matr{H}_i =  \matr{\Pi}^{l_i} \matr{\Delta}^{\kappa_i}$, then, the the $m$-th row and $n$-th column of $\matr{H}_i $ is
\begin{align}
H_i[m, n] &=  \sum_{q=0}^{N-1}    {\Pi}^{l_i}[m,q] {\Delta}^{\kappa_i}[q,n]\\
	    &=  \sum_{q=0}^{N-1}  \delta[m-q-l_i] e^{j\frac{2\pi \kappa_i}{N}q} \delta[q-n]\\
	    &=  \delta[m-n-l_i] e^{j\frac{2\pi \kappa_i}{N}n}. \label{eq:onePathofH}
\end{align}

Using \eqref{eq:fdnt} and \eqref{eq:onePathofH}, we can expand ${H}_c \left[m,n\right] $ as 
\begin{align}
&{H}_c \left[m,n\right]  \\
&\begin{aligned}=  \sum_{i_1=0}^{N-1}\sum_{i_2=0}^{N-1}  \sum_{i=0}^{L-1}& \frac{1}{N} e^{-j \frac \pi 4} e^{j \frac{\pi (m-i_1)^2}{N}}
						h_i \\&\delta[i_1-i_2-l_i] e^{j2\pi \frac{\kappa_i i_2}{N}}  e^{j \frac \pi 4} e^{-j \frac{\pi (i_2-n)^2}{N}} \end{aligned} \\
						&\begin{aligned} =\frac 1N \sum_{i_1=0}^{N-1}\sum_{i_2=0}^{N-1}  \sum_{i=0}^{L-1}  & e^{j \frac{\pi (m-i_1)^2}{N}}
						h_i \delta[i_1-i_2-l_i] \\& e^{j2\pi \frac{\kappa_i i_2}{N}}  e^{-j \frac{\pi (i_2-n)^2}{N}} \end{aligned} \\
						&\begin{aligned} = \frac 1N \sum_{i_1=0}^{N-1}\sum_{i_2=0}^{N-1}  \sum_{i=0}^{L-1} &h_i \delta[i_1-i_2-l_i] 
						  \\& e^{j \frac{\pi}{N} \left( m^2 -n^2 + i_1^2 -i_2^2 + 2i_2(n + \kappa_i) - 2i_1m\right)} \end{aligned}\\
						&\stackrel{(a)}{=}  \frac VN  \sum_{i_1=0}^{N-1} \sum_{i=0}^{L-1} h_i 
						  e^{j \frac{\pi}{N} \left( 2 i_1(n-m + l_i +\kappa_i) -l_i^2 - 2nl_i - 2l_i\kappa_i  \right)} \\
						&= \frac VN  \sum_{i=0}^{L-1} h_i e^{-j \frac{\pi}{N}(l_i^2 +2nl_i + 2l_i\kappa_i )}
						\sum_{i_1=0}^{N-1}   e^{j 2\frac{\pi}{N} i_1\left(  n-m + l_i +\kappa_i\right)} \\
						&\stackrel{(b)}{=}  V \sum_{i=0}^{L-1} h_i e^{-j \frac{\pi}{N}(l_i^2 +2nl_i + 2l_i\kappa_i )} 
						\delta[n-m+l_i +\kappa_i \Mod{N}] \label{eq:ocdm_Ris1}\\ 
						&\stackrel{(c)}{=}   \sum_{i=0}^{L-1} h_i e^{j \frac{\pi}{N}( 2m\kappa_i - 2l_i\kappa_i - \kappa_i^2)} 
						\delta[n-m+l_i +\kappa_i \Mod{N}].
\end{align}
For simplicty, we let $V = e^{j \frac{\pi}{N} (m^2-n^2)}$. In step (a), we take $i_2 = i_1-l_i$; in (b), we claim that the second summation is either zero or $N$. That is, if $n-m+l_i +\kappa_i = 0\Mod{N}$, then $\sum_{i_1=0}^{N-1}   e^{j 2\frac{\pi}{N} i_1\left(  n-m + l_i +\kappa_i\right)} =N$, otherwise the sum is zero. 

Moreover, we simplify the term in (c), by taking $n= m-l_i -\kappa_i  \Mod{N}$.

\section{Proof of Theorem \ref{Th:time2Freq}}\label{apen:timetoFreq}
\noindent We wish to find a finite-length cyclically CAZAC sequence of length $N$,  $\left\{a[k]\right\}_{k=0}^{N-1}$, such that a unit cyclical time shift is equivalent, up to a multiplicative factor, $\beta$, to the same sequence modulated in frequency by $\frac PN$. That is, 
\begin{equation}
a[k-1\bmod N] =\beta a[k] e^{\frac{j 2 \pi kP}{N}}, \quad 0 \leq k<N. \label{eq:time2freq}
\end{equation}

The projection of any arbitrary circular time shifts onto frequency shifts can be generalized as follows: 
\begin{lemma}\label{lem:mtimeshift}
Given \eqref{eq:time2freq}, where a single unit of time shift leads to a frequency shift of $\frac PN$ units, then, a time shift of $m, m \in \mathbb{Z}_{\neq 0}$ leads to a $m \frac PN$ frequency shift and a constant multiplicative phase shift of $e^{-j\pi\frac{m(m-1)\tau}{N}}$. In other words, we get
\begin{equation}
a[k-m\bmod N] =\beta^m a[k] e^{\frac{j 2 \pi mk P}{ N}} e^{-j\pi\frac{m(m-1)P}{N}}. \label{eq:Timeseq_m}
\end{equation}

\end{lemma}
\begin{proof}
 We use proof by induction. The base case for $m = 1$ is easily verifiable. 
Now, we form the hypothesis that for $m= n$, we get 
 \begin{equation}
a[k-n\bmod N] =\beta^n a[k] e^{\frac{j 2 \pi nk P}{ N}} e^{-j\pi\frac{n(n-1)P}{N}}. 
\end{equation}

\noindent In the induction step, we work with $m= n+1$. Our goal is to prove that 
 \begin{equation}
a[k-(n+1)\bmod N] = \beta^{(n+1)} a[k]  e^{\frac{j 2 \pi k(n+1)P }{ N}} e^{-j\pi\frac{n(n+1)P}{N}}. 
\end{equation}

\noindent We start with our hypothesis and let $k = k-1$, that is
\begin{align}
a[k-1-n\bmod N] =\beta^n a[k-1] e^{\frac{j 2 \pi n(k-1) P}{ N}} e^{-j\pi\frac{n(n-1)P}{N}}.
\end{align}
By substituting $a[k-1\bmod N]$ on the right-hand side by the base case \eqref{eq:time2freq}, we get 
\begin{align}
\begin{split}
&a[k-(n+1)\bmod N]\\&=\beta^n \left(\beta a[k] e^{\frac{j 2 \pi kP}{N}} \right) e^{\frac{j 2 \pi n(k-1) P}{ N}} e^{-j\pi\frac{n(n-1)P}{N}}, 
\end{split}
\end{align}
which we can easily simplify to obtain 

\begin{align}
a[k-(n+1)\bmod N] = \beta^{(n+1)} a[k]  e^{\frac{j 2 \pi k(n+1)P }{ N}} e^{-j\pi\frac{n(n+1)P}{N}}. 
\end{align}

\end{proof}

\noindent It is easy to see that for forward time-shifts, \eqref{eq:Timeseq_m} becomes 
\begin{equation}
a[k+m\bmod N] = \beta^{-m} a[k] e^{-\frac{j 2 \pi mkP}{N}} e^{-j\pi\frac{m(m+1)P}{N}}. \label{eq:Time_negm}
\end{equation}

We can draw the conclusion that if an integer circular time shift of $1$-unit leads to $\frac PN$-unit frequency shifts, the $m$-unit time shift will lead to $m$ multiples of $\frac PN$  frequency shifts and a constant multiplicative factor.  
This sequence is very interesting because once we know the value of $a[0]$, we can deduce the value of $a[k], 0<k<N$, scaled by a frequency shift and a multiplicative factor. From  \eqref{eq:Time_negm}, we have 
\begin{equation}
a[m\bmod N] = \beta^{-m}  e^{-j\pi\frac{m^2P}{N}} e^{-j\pi\frac{mP}{N}} a[0], \label{eq:s_mSeq}
\end{equation}
and due to periodicity of the sequence $a[k]$, we can also express this as 
\begin{equation}
\begin{split}
 a[0] &= a[N] = \left(  \beta^{-1} e^{-j\frac{\pi P}{N}}  \right)^{N} e^{-j\pi P N} a[0] \\&= \left(  \beta^{-1} e^{-j\frac{\pi P}{N}} \right)^{N} (-1)^{P N} a[0],\end{split}
\end{equation}
which perpetuates that $ \left(  \beta^{-1} e^{-j\frac{\pi P}{N}} \right)(-1)^{P} = e^{j\frac{2\pi t}{N}}$, for some $t$, where $0\leq t<N$. 
 Hence, we can rewrite \eqref{eq:s_mSeq} as
 \begin{equation}
a[k]=(-1)^{Pk} e^{j  \frac{2 \pi  t k}{N}} e^{-j\frac{\pi k^{2}P}{N}} a[0], \quad 0 \leq k<N. \label{eq:sk_Freq}
\end{equation}

Now that we have the expression for $a[k]$, we can find the autocorrelation---as follows
\begin{equation}
\begin{split}
\Gamma_{v} & =\sum_{k=0}^{N-1} a[k]^{\ast} a[k+v]
 \\&=(-1)^{Pv} e^{j \frac{ 2 \pi  tv}{N}} e^{-j\frac{\pi P v^2}{N}}  \lvert a[0]\rvert^{2} \left(\sum_{k=0}^{N-1} e^{-j \frac{2\pi P kv}{N}}\right). \end{split}\label{eq:autocorre_case2}
\end{equation}

\noindent It is easy to see from \eqref{eq:autocorre_case2} that 
\begin{equation}
\Gamma_{0}=N\lvert a[0]\rvert^{2}=N,
\end{equation}
which implies that $a[0]$ is a pure phase. That is, the modulus of $a[0]$  is unity. This value can be chosen arbitrarily without any loss of generality. On the other hand, to make the autocorrelation zero for $0<v<N$, $kv$ must be non-divisible by $N$. Thus, it is necessary and sufficient that $P$ be coprime with $N$. 
 We can simplify \eqref{eq:sk_Freq} to
\begin{equation}
a[k] 
= e^{-j \pi Pk \frac{k-N- 2q}{N}} a[0], \label{eq:sk_step6_case2}
\end{equation}
where we take $qb = t \Mod{N}$, for some integer $q$, which facilitates the factorization of the exponents. The term $-2 q-N$ characterizes the parity of $N$; we choose to write the parity term as $c+2 m$ for some integer $m$, and  $c=N\Mod{2}$. We can rewrite \eqref{eq:sk_step6_case2} as

\begin{equation}
a[k] = e^{-j \pi P k \frac{k+c+2 m}{N}} a[0]. 
\end{equation}

\noindent Without loss of generality, we can take $a[0]=1$. Then, we have

\begin{equation}
a[k]=e^{-j \pi P k \frac{k+c+2 m}{N}}, 
\end{equation}
which gives us the general expressions of a Zadoff-Chu sequence. In conclusion, we can say that Zadoff-Chu sequences are the only sequences that convert a time shift into a frequency shift.

\bibliographystyle{IEEEtran}
\bibliography{IEEEabrv, otfs}

\end{document}